\documentclass{CSML}
\pdfoutput=1

\usepackage{lastpage}
\newcommand{\dOi}{14(1:11)2018}
\lmcsheading{}{1--\pageref{LastPage}}{}{}%
{May~31, 2016}{Jan.~30, 2018}{}

\usepackage{hyperref}
\hypersetup{hidelinks}
\usepackage{amssymb}
\usepackage{bussproofs}
\usepackage{stmaryrd}
\usepackage{fancyvrb,courier}
\usepackage{tabularx}

\usepackage{mdframed}

\usepackage{xargs}

\usepackage{macros}
\begin{document}

\title[Logical relations for coherence of effect subtyping]
      {Logical relations for coherence of effect subtyping}

\author[Biernacki]{Dariusz Biernacki}
\address{Institute of Computer Science, University of
  Wroc{\l}aw\\ Joliot-Curie 15, 50-383 Wroc{\l}aw, Poland}
\email{dabi@cs.uni.wroc.pl}

\author[Polesiuk]{Piotr Polesiuk}
\address{\vskip-6pt}
\email{ppolesiuk@cs.uni.wroc.pl}


\titlecomment{An extended and revised version
  of~\cite{Biernacki-Polesiuk:TLCA15}. This work has been supported by
  National Science Centre, Poland, grant no. 2011/03/B/ST6/00348 and
  grant no. 2014/15/B/ST6/00619.}

%%%%%%%%%%%%%%%%%%%%%%%%%%%%%%%%%%%%%%%%%%%%%%%%%%%%%%%%%%%%%%%%%%%%%%%%%%%

\begin{abstract}
  A coercion semantics of a programming language with subtyping is
  typically defined on typing derivations rather than on typing
  judgments. To avoid semantic ambiguity, such a semantics is expected
  to be coherent, i.e., independent of the typing derivation for a
  given typing judgment. In this article we present heterogeneous,
  biorthogonal, step-indexed logical relations for establishing the
  coherence of coercion semantics of programming languages with
  subtyping. To illustrate the effectiveness of the proof method, we
  develop a proof of coherence of a type-directed, selective CPS
  translation from a typed call-by-value lambda calculus with
  delimited continuations and control-effect subtyping. The article is
  accompanied by a Coq formalization that relies on a novel shallow
  embedding of a logic for reasoning about step-indexing.
\end{abstract}

\maketitle

%%%%%%%%%%%%%%%%%%%%%%%%%%%%%%%%%%%%%%%%%%%%%%%%%
\section{Introduction}\label{sec:intro}

Programming languages that allow for subtyping, i.e., a mechanism
facilitating coercions of expressions of one type to another, are
usually given either a subset semantics, where one type is considered
a subset of another type, or a coercion semantics, where expressions
are explicitly converted from one type to another. In the presence of
subtyping, typing derivations depend on the occurrences of the
subtyping judgments and, therefore, typing judgments do not have
unique typing derivations. Consequently, a coercion semantics that
interprets subtyping judgments by introducing explicit type coercions
is defined on typing derivations rather than on typing judgments. But
then a natural question arises as to whether such a semantics is
coherent, i.e., whether it does not depend on the typing derivation.

The problem of coherence has been considered in a variety of typed
lambda calculi. Reynolds proved the coherence of the denotational
semantics for intersection types in the category-theoretic
setting~\cite{Reynolds:TACS91}. Breazu-Tannen et al. proved the
coherence of a coercion translation from the lambda calculus with
polymorphic, recursive and sum types to system F, by showing that any
two derivations of the same judgment are translated to provably equal
terms in the target calculus~\cite{Breazu-Tannen-al:IaC91}. Curien and
Ghelli introduced a translation from system F$_\leq$ to a calculus
with explicit coercions and showed that any two derivations of the
same judgment are translated to terms that are normalizable to a
unique normal form~\cite{Curien-Ghelli:MSCS92}. Finally, Schwinghammer
followed Curien and Ghelli's approach to prove the coherence of
coercion translation from Moggi's computational lambda calculus with
subtyping, except that he normalizes derivations in a
semantics-preserving way, rather than terms in a dedicated calculus of
coercions~\cite{Schwinghammer:JFP09}. Schwinghammer's presentation is
akin to Mitchell's for the simply-typed lambda
calculus~\cite[Chapter~10]{Mitchell:96}.

The results listed above fall into two categories: those that hinge on
the existence of a common subtype of two different types for the same
term~\cite{Reynolds:TACS91,Breazu-Tannen-al:IaC91}, and those that
rely on finding a normal form for a representation of the derivation
and hinge on showing that such normal forms are unique for a given
typing
judgment~\cite{Curien-Ghelli:MSCS92,Schwinghammer:JFP09,Mitchell:96}.
When the source calculus under consideration is presented in the
spirit of the lambda calculus {\`a} la Church, i.e., the lambda
abstractions are type annotated, as is the case in all the
aforementioned articles that follow the normalization-based approach,
the term and the typing context indeed determine the shape of the
normal derivation (modulo a top level coercion that depends on the
type of the term)~\cite[Chapter~10]{Mitchell:96}. However, in calculi
{\`a} la Curry this is no longer the case and the method cannot be
directly applied. Still, if the calculus is at least weakly
normalizing, one can hope to recover the uniqueness property for
normal typing derivations for source terms in normal form, assuming
that term normalization preserves the coercion semantics. For
instance, in the simply typed $\lambda$-calculus the typing context
uniquely determines the type of the term in the function position in
applications building a $\beta$-normal form, and, hence, derivations
in normal form for such terms are unique. This line of reasoning
cannot be used when the calculus includes recursion. Similarly, the
lambda calculus {\`a} la Curry (and other systems extending it) does
not in general satisfy the property of common subtype.

In this article, we consider the coherence problem in calculi where
none of the existing techniques can be directly applied. The coercion
semantics we study translate typing derivations in the source calculus
to a corresponding target calculus with explicit type coercions (that
in some cases can be further replaced with equivalent lambda-term
representations) and our criterion for coherence of the translation is
contextual equivalence~\cite{JHMorris:PhD} in the target calculus.

The main result of this work is a construction of logical relations
for establishing such a notion of coherence of coercion semantics,
applicable in a variety of calculi. In particular, we address the
problem of coherence of a type-directed CPS (continuation-passing
style) translation from the call-by-value $\lambda$-calculus with
delimited-control operators and control-effect subtyping introduced by
Materzok and the first author~\cite{Materzok-Biernacki:ICFP11},
extended with recursion. While the translation for the calculus with
explicit type annotations has been shown to be coherent in terms of an
equational theory in a target calculus~\cite{Materzok:CSL13}, no CPS
coercion translation for the original version, let alone extended with
recursion, has been proven coherent.

The reasons why coherence in this calculus is important are
twofold. First of all, it is very expressive and therefore interesting
from the theoretical point of view. In particular, the calculus has
been shown to generalize the canonical type-and-effect system for
Danvy and Filinski's $\mathsf{shift}$ and $\mathsf{reset}$ control
operators~\cite{Danvy-Filinski:DIKU89,Danvy-Filinski:LFP90}, and,
furthermore, that it is strictly more expressive than the CPS
hierarchy of Danvy and
Filinski~\cite{Materzok-Biernacki:APLAS12}. These results heavily rely
on the effect subtyping relation that, e.g., allows to coerce pure
expressions (i.e., control-effect free) to effectful ones. From a more
practical point of view, the selective CPS translation, that leaves
pure expressions in direct style and introduces explicit coercions to
interpret effect subtyping in the source calculus, is a good candidate
for embedding the control operators in an existing programming
language, such as Scala~\cite{Rompf-al:ICFP09}.

In order to deal with the complexity of the source calculus and of the
translation itself, we introduce binary logical relations on terms of
the target calculus that are: heterogeneous,
biorthogonal~\cite{Krivine:APAL94,Pitts-Stark:HOOTS98,Dreyer-al:JFP12},
and
step-indexed~\cite{Appel-McAllester:TOPLAS01,Ahmed:ESOP06,Ahmed-al:POPL09}.
Heterogeneity allows us to relate terms of different types, and in
particular those in continuation-passing style with those in direct
style. This is a crucial property, since the same term can have a pure
type, resulting in a direct-style term through the translation and
another, impure type, resulting in a term in continuation-passing
style. Relating such terms requires quantification over types and to
assure well-foundedness of the construction, we need to use
step-indexing, which also supports reasoning about recursion, even if
not in a critical way. We follow Dreyer et al.~\cite{Dreyer-al:LMCS11}
in using logical step-indexed logical relations in our presentation of
step-indexing. Biorthogonality, by imposing a particular order of
evaluation on expressions, simplifies the construction of the logical
relations. It also facilitates reasoning about continuations
represented as evaluation contexts.

Apart from the calculus with effect subtyping, we have used the ideas
presented in this article to show coherence of subtyping in several
other calculi, including the simply typed lambda calculus with
subtyping~\cite[Chapter~10]{Mitchell:96} extended with recursion, the
calculus of intersection types~\cite{Reynolds:TACS91}, and the lambda
calculus with subtyping and the control operator~$\mathsf{call/cc}$.

The article is accompanied by a Coq development containing a library
IxFree that provides a new shallow embedding of the logic for
reasoning about step-indexed logical relations, and a complete
formalization of the proofs presented in the rest of the article. The
code is available
at~\href{https://bitbucket.org/pl-uwr/coherence-logrel}{https://bitbucket.org/pl-uwr/coherence-logrel}.

The rest of this article is structured as follows. In
Section~\ref{sec:stepIndexedLogicalRelations}, we briefly present
Dreyer et al.'s logic for reasoning about step
indexing~\cite{Dreyer-al:LMCS11} on which we base our presentation. In
Section~\ref{sec:introducingLogicalRelations}, we introduce the
construction of the logical relations in a simple yet sufficiently
interesting scenario---the simply typed lambda calculus {\`a} la Curry
with natural numbers, type $\ttop$, general recursion and standard
subtyping. The goal of this section is to introduce the basic
ingredients of the proof method before embarking on a considerably
more challenging journey in the subsequent section. In
Section~\ref{sec:coherenceEffectSubtyping}, we present the main result
of the article---the logical relations for establishing the coherence
of the CPS translation from the calculus of delimited control with
effect subtyping. In Section~\ref{sec:formalization}, we describe the
main ideas behind our Coq formalization. In
Section~\ref{sec:conclusion}, we summarize the article.

%%%%%%%%%%%%%%%%%%%%%%%%%%%%%%%%%%%%%%%%%%%%%%%%%
\section{Reasoning about step-indexed logical relations}
\label{sec:stepIndexedLogicalRelations}

Step-indexed logical
relations~\cite{Appel-McAllester:TOPLAS01,Ahmed:ESOP06,Ahmed-al:POPL09}
are a powerful tool for reasoning about programming languages. Instead
of describing a general behavior of program execution, they focus on
the first $n$ computation steps, where the step index $n$ is an
additional parameter of the relation. This additional parameter makes
it possible to define logical relations inductively not only on the
structure of types, but also on the number of computation steps that
are allowed for a program to make and, therefore, they provide an
elegant way to reason about features that introduce non-termination to
the programming language, including recursive
types~\cite{Ahmed:ESOP06} and references~\cite{Ahmed-al:POPL09}.

However, reasoning directly about step-indexed logical relations is
tedious because proofs become obscured by step-index arithmetic.
Dreyer et al. \cite{Dreyer-al:LMCS11} proposed logical step-indexed
logical relations (LSLR) to avoid this problem.  The LSLR logic is an
intuitionistic logic for reasoning about one particular Kripke model:
where possible worlds are natural numbers (step-indices) and where
future worlds have smaller indices than the present one. All formulas
are interpreted as monotone (non-increasing) sequences of truth
values, whereas the connectives are interpreted as usual. In
particular, in the case of implication we quantify over all future
worlds to ensure monotonicity, so the formula $\varphi\Rightarrow\psi$
is valid at index $n$ (written $n\models\varphi\Rightarrow\psi$) iff
$k\models\varphi$ implies $k\models\psi$ for every $k\leq n$.  In
contrast to Dreyer et al. we do not assume that all formulas are valid
in world $0$, because it is not necessary.

The LSLR logic is also equipped with a modal operator $\later$
(later), to provide access to strictly future worlds. The formula
$\later\varphi$ means \emph{$\varphi$ holds in any future world}, or
formally $\later\varphi$ is always valid at world $0$, and
$n+1\models\later\varphi$ iff $\varphi$ is valid at $n$ (and other
future worlds by monotonicity).  The later operator comes with two
inference rules:

\vspace{2mm}
\begin{center}
\AxiomC{$\Gamma,\Sigma\vdash\varphi$}
\rulelabel{$\later$-intro}
\UnaryInfC{$\Gamma,\later\Sigma\vdash\later\varphi$}
\DisplayProof
\qquad
\AxiomC{$\Gamma,\later\varphi\vdash\varphi$}
\rulelabel{L{\"o}b}
\UnaryInfC{$\Gamma\vdash\varphi$}
\DisplayProof
\end{center}

\vspace{2mm}\noindent
The first rule allows one to shift reasoning to a future world, making
the assumptions about the future world available.  The L{\"o}b rule
expresses an induction principle for indices. Note that the premise of
the rule also captures the base case, because the assumption
$\later\varphi$ is trivial in the world $0$. The later operator comes
with no general elimination rule.

Predicates in LSLR logic as well as step-indexed logical relations can
be defined inductively on indices.  More generally, we can define a
recursive predicate $\mu{}r.\varphi(r)$, provided all occurrences of
$r$ in $\varphi$ are guarded by the later operator, to guarantee
well-foundedness of the definition. For the sake of readability, in
this paper we define recursive predicates and relations by giving a
set of clauses instead of using the $\mu$ operator.

Since the logic is developed for reasoning about one particular model,
we can freely add new inference rules for the logic if we prove they
are valid in the model.  We can also add new relations or predicates
to the logic if we provide their monotone interpretation.  In
particular, constant functions are monotone, so we can safely use
predicates defined outside of the logic, such as typing or reduction
relations.

\begin{figure}[!!t]
  \begin{mdframed}
  \begin{tabularx}{\textwidth}{rcXr}
    $\tau$ & $\bnfdef$ & $\tnat \bnfor \ttop \bnfor \tarrow{\tau}{\tau}$
    & (types) \\[1mm]
    $e$ & $\bnfdef$ & $\var{x} \bnfor \lam{x}{e} \bnfor \app{e}{e}
    \bnfor \fix{x}{x}{e} \bnfor \natconst{n}$
    & (expressions)
  \end{tabularx}
  \vspace{4mm}
  \hrule
  \vspace{4mm}
      {\begin{center}
          \AxiomC{\phantomPremise}
          \rulelabel{\ruleSourceSRefl}
          \UnaryInfC{$\jsubtype{\tau}{\tau}$}
          \DisplayProof
          \quad
          \AxiomC{$\jsubtype{\tau_2}{\tau_3}$}
          \AxiomC{$\jsubtype{\tau_1}{\tau_2}$}
          \rulelabel{\ruleSourceSTrans}
          \BinaryInfC{$\jsubtype{\tau_1}{\tau_3}$}
          \DisplayProof
          \quad
          \AxiomC{\phantomPremise}
          \rulelabel{\ruleSourceSTop}
          \UnaryInfC{$\jsubtype{\tau}{\ttop}$}
          \DisplayProof
          \\[3mm]
          \AxiomC{$\jsubtype{\tau_2'}{\tau_1'}$}
          \AxiomC{$\jsubtype{\tau_1}{\tau_2}$}
          \rulelabel{\ruleSourceSArrow}
          \BinaryInfC{$\jsubtype{(\tarrow{\tau_1'}{\tau_1})}{
              (\tarrow{\tau_2'}{\tau_2})}$}
          \DisplayProof
          \quad
          \AxiomC{$(x:\tau)\in\Gamma$}
          \rulelabel{\ruleSourceTVar}
          \UnaryInfC{$\jtyping{\Gamma}{\var{x}}{\tau}$}
          \DisplayProof
          \quad
          \AxiomC{$\jtyping{\eextend{\Gamma}{x}{\tau_1}}{e}{\tau_2}$}
          \rulelabel{\ruleSourceTAbs}
          \UnaryInfC{$\jtyping{\Gamma}{\lam{x}{e}}{\tarrow{\tau_1}{\tau_2}}$}
          \DisplayProof
          \\[3mm]
          \AxiomC{$\jtyping{\Gamma}{e_1}{\tarrow{\tau_2}{\tau_1}}$}
          \AxiomC{$\jtyping{\Gamma}{e_2}{\tau_2}$}
          \rulelabel{\ruleSourceTApp}
          \BinaryInfC{$\jtyping{\Gamma}{\app{e_1}{e_2}}{\tau_1}$}
          \DisplayProof
          \quad
          \AxiomC{$\jtyping{
              \eextend{\eextend{\Gamma}{f}{\tarrow{\tau_1}{\tau_2}}}{x}{\tau_1}
            }{e}{\tau_2}$}
          \rulelabel{\ruleSourceTFix}
          \UnaryInfC{$\jtyping{\Gamma}{\fix{f}{x}{e}}{\tarrow{\tau_1}{\tau_2}}$}
          \DisplayProof
          \\[5mm]
          \AxiomC{\phantomPremise}
          \rulelabel{\ruleSourceTConst}
          \UnaryInfC{$\jtyping{\Gamma}{\natconst{n}}{\tnat}$}
          \DisplayProof
          \quad
          \AxiomC{$\jtyping{\Gamma}{e}{\tau}$}
          \AxiomC{$\jsubtype{\tau}{\tau'}$}
          \rulelabel{\ruleSourceTSub}
          \BinaryInfC{$\jtyping{\Gamma}{e}{\tau'}$}
          \DisplayProof
      \end{center}}
  \end{mdframed}
      \caption{\label{fig:stlcSource}The source
        language---the $\lambda$-calculus with subtyping}
\end{figure}

%%%%%%%%%%%%%%%%%%%%%%%%%%%%%%%%%%%%%%%%%%%%%%%%%
\section{Introducing the logical relations}
\label{sec:introducingLogicalRelations}

In this section we prove the coherence of subtyping in the
simply-typed call-by-value lambda calculus extended with recursion,
where the coercion semantics is given by a standard translation to the
simply-typed lambda calculus with explicit
coercions~\cite{Curien-Ghelli:MSCS92}. Our goal here is to introduce
the proof method in a simple scenario, so that in
Section~\ref{sec:coherenceEffectSubtyping} we can focus on issues
specific to control effects. The logical relations we present in this
section are biorthogonal and step-indexed, which is not strictly
necessary but it makes the development more elegant. Furthermore,
biorthogonality and step-indexing become crucial in handling more
complicated calculi such as the one of
Section~\ref{sec:coherenceEffectSubtyping} and, therefore, are
essential for the method to scale.

%%%%%%%%%%%%%%%%%%%%%%%%
\subsection{The simply-typed lambda calculus with subtyping}
\label{subsec:STLC}

The syntax and typing rules for the source language are given in
Figure~\ref{fig:stlcSource}. The language is the simply-typed lambda
calculus extended with recursive functions ($\fix{f}{x}{e}$) and
natural numbers ($\natconst{n}$). For clarity of the presentation we
do not consider any primitive operations on natural numbers, but they
could be seamlessly added to the language. Extending the language with
additional basic types is a little bit more subtle, as discussed in
Section~\ref{subsubsec:multiBaseTypes}. We include the type $\ttop$,
to make the subtyping relation interesting. The typing and subtyping
rules are standard~\cite[Chapter~10]{Mitchell:96}, where the typing
environment $\Gamma$ associates variables with their types and is
represented as a finite set of such pairs, noted $(x:\tau)$. In the
rest of the article we assume the standard notions and conventions
concerning variable binding and $\alpha$-conversion of
terms~\cite{Barendregt:84}.

%%%%%%%%%%%%%%%%%%%%%%%%
\subsection{Coercion semantics}
\label{subsec:coercionSemanticsSTLC}

The semantics of the source language is given by a translation of the
typing derivations to a target language that extends the source
language with explicit type coercions (and replaces $\ttop$ with
$\tunit$).

%%%%%%%%%%%%
\subsubsection{Target calculus}
\label{subsubsec:targetSTLC}

The coercions express conversion of a term from one type to another,
according to the subtyping relation. Figure~\ref{fig:stlcTarget}
contains syntax, typing rules and reduction rules of the target
language. The type coercions $c$ and their typing rules correspond
exactly to the subtyping rules of the source language.  The grammar of
terms contains explicit coercion application of the form
$\capp{c}{e}$.

The operational semantics of the target language takes the form of the
reduction semantics, where terms are decomposed into an evaluation
context and a redex. We use the standard notation $\inctx{E}{e}$ for
plugging the term $e$ into the context $E$, and
similarly---$\ctxcomp{E}{E'}$ for plugging the context $E'$ in $E$,
i.e., for context composition. The grammar of evaluation contexts
extends the standard call-by-value $\lambda$-calculus contexts with
contexts of the form $\crcectx{c}{E}$ that enforce evaluating an
argument of a coercion before the actual conversion takes place. It
can be shown that the reduction relation of
Figure~\ref{fig:stlcTarget} is deterministic.

\begin{figure}[!!t]
  \begin{mdframed}
    \begin{tabularx}{\textwidth}{rcXr}
      $\tau$ & $\bnfdef$ & $\tnat \bnfor \tunit \bnfor \tarrow{\tau}{\tau}$
      & (types) \\[1mm]
      $c$ & $\bnfdef$ & $\cid \bnfor \ccomp{c}{c} \bnfor \ctop \bnfor \carrow{c}{c}$
      & (coercions) \\[1mm]
      $e$ & $\bnfdef$ & $\var{x} \bnfor \lam{x}{e} \bnfor \app{e}{e}
      \bnfor \capp{c}{e} \bnfor \fix{x}{x}{e} 
      \bnfor \natconst{n}
      \bnfor \vunit$
      & (expressions) \\[1mm]
      $v$ & $\bnfdef$ & $\var{x} \bnfor \lam{x}{e} \bnfor \fix{x}{x}{e}
      \bnfor \capp{(\carrow{c}{c})}{v} \bnfor \natconst{n} \bnfor \vunit$
      & (values) \\[1mm]
      $E$ & $\bnfdef$ & $\mtectx \bnfor \argectx{E}{e} \bnfor \valectx{v}{E}
      \bnfor \crcectx{c}{E}
      $
      & (evaluation contexts)
    \end{tabularx}
    \vspace{4mm}
    \hrule
    \vspace{4mm}
           {
             \begin{center}
               \AxiomC{\phantomPremise}
               \rulelabel{\ruleTargetSRefl}
               \UnaryInfC{$\jctyping{\cid}{\tau}{\tau}$}
               \DisplayProof
               \quad\quad
               \AxiomC{$\jctyping{c_1}{\tau_2}{\tau_3}$}
               \AxiomC{$\jctyping{c_2}{\tau_1}{\tau_2}$}
               \rulelabel{\ruleTargetSTrans}
               \BinaryInfC{$\jctyping{\ccomp{c_1}{c_2}}{\tau_1}{\tau_3}$}
               \DisplayProof
               \\[5mm]
               \AxiomC{\phantomPremise}
               \rulelabel{\ruleTargetSTop}
               \UnaryInfC{$\jctyping{\ctop}{\tau}{\tunit}$}
               \DisplayProof
               \quad\quad
               \AxiomC{$\jctyping{c_1}{\tau_2'}{\tau_1'}$}
               \AxiomC{$\jctyping{c_2}{\tau_1}{\tau_2}$}
               \rulelabel{\ruleTargetSArrow}
               \BinaryInfC{$\jctyping{\carrow{c_1}{c_2}}{(\tarrow{\tau_1'}{\tau_1})}{
                   (\tarrow{\tau_2'}{\tau_2})}$}
               \DisplayProof
               \\[5mm]
               \AxiomC{$(x:\tau)\in\Gamma$}
               \rulelabel{\ruleTargetTVar}
               \UnaryInfC{$\jtyping{\Gamma}{\var{x}}{\tau}$}
               \DisplayProof
               \quad\quad
               \AxiomC{$\jtyping{\eextend{\Gamma}{x}{\tau_1}}{e}{\tau_2}$}
               \rulelabel{\ruleTargetTAbs}
               \UnaryInfC{$\jtyping{\Gamma}{\lam{x}{e}}{\tarrow{\tau_1}{\tau_2}}$}
               \DisplayProof
               \\[5mm]
               \AxiomC{$\jtyping{\Gamma}{e_1}{\tarrow{\tau_2}{\tau_1}}$}
               \AxiomC{$\jtyping{\Gamma}{e_2}{\tau_2}$}
               \rulelabel{\ruleTargetTApp}
               \BinaryInfC{$\jtyping{\Gamma}{\app{e_1}{e_2}}{\tau_1}$}
               \DisplayProof
               \quad\quad
               \AxiomC{$\jctyping{c}{\tau}{\tau'}$}
               \AxiomC{$\jtyping{\Gamma}{e}{\tau}$}
               \rulelabel{\ruleTargetTCApp}
               \BinaryInfC{$\jtyping{\Gamma}{\capp{c}{e}}{\tau'}$}
               \DisplayProof
               \\[5mm]
               \AxiomC{$\jtyping{
                   \eextend{\eextend{\Gamma}{f}{\tarrow{\tau_1}{\tau_2}}}{x}{\tau_1}
                 }{e}{\tau_2}$}
               \rulelabel{\ruleTargetTFix}
               \UnaryInfC{$\jtyping{\Gamma}{\fix{f}{x}{e}}{\tarrow{\tau_1}{\tau_2}}$}
               \DisplayProof
               \quad\quad
               \AxiomC{\phantomPremise}
               \rulelabel{\ruleTargetTConst}
               \UnaryInfC{$\jtyping{\Gamma}{\natconst{n}}{\tnat}$}
               \DisplayProof
               \\[5mm]
               \AxiomC{\phantomPremise}
               \rulelabel{\ruleTargetTTop}
               \UnaryInfC{$\jtyping{\Gamma}{\vunit}{\tunit}$}
               \DisplayProof
           \end{center}}
           \vspace{4mm}
           \hrule
           \vspace{4mm}
           \begin{center}
             \hspace{-3mm}
             \begin{tabular}{cc}
               \setlength{\tabcolsep}{1pt}
               \begin{tabular}[t]{rcl}
                 $\inctx{E}{\app{(\lam{x}{e})}{v}}$ & $\rawredi{\beta}$ & 
                 $\inctx{E}{\subst{e}{x}{v}}$ \\[2mm]
                 $\inctx{E}{\app{(\fix{f}{x}{e})}{v}}$ & $\rawredi{\beta}$ &
                 $\inctx{E}{\bisubst{e}{f}{\fix{f}{x}{e}}{x}{v}}$
               \end{tabular}
               &
               \setlength{\tabcolsep}{1pt}
               \begin{tabular}[t]{rcl}
                 $\inctx{E}{\capp{\cid}{v}}$ & $\rawredi{\iota}$ & $\inctx{E}{v}$ \\[2mm]
                 $\inctx{E}{\capp{(\ccomp{c_1}{c_2})}{v}}$ & $\rawredi{\iota}$ &
                 $\inctx{E}{\capp{c_1}{(\capp{c_2}{v})}}$ \\[2mm]
                 $\inctx{E}{\capp{\ctop}{v}}$ & $\rawredi{\iota}$ &
                 $\inctx{E}{\vunit}$ \\[2mm]
                 $\inctx{E}{\app{\capp{(\carrow{c_1}{c_2})}{v_1}}{v_2}}$ & $\rawredi{\iota}$ &
                 $\inctx{E}{\capp{c_2}{(\app{v_1}{(\capp{c_1}{v_2})})}}$
               \end{tabular}
             \end{tabular}
           \end{center}
  \end{mdframed}
  \caption{\label{fig:stlcTarget}The target language---the
    $\lambda$-calculus with explicit coercions}
\end{figure}

The semantics distinguishes between $\beta$-rules that perform actual
computations and $\iota$-rules that rearrange coercions. Both of them
are used during program evaluation. We say that program $e$ terminates
(written $\stops{e}$) when it can be reduced to a value using both
sorts of reduction rules, according to the evaluation strategy
determined by the evaluation contexts.

The principle behind the operational semantics of coercions, given by
the $\iota$-rules, is to structurally reduce complex coercions
$\ccomp{c_1}{c_2}$ and $\carrow{c_1}{c_2}$ to their subcoercions $c_1$
and $c_2$, until one of the two basic coercions $\cid$ or $\ctop$ is
reached and can be trivially applied to perform the actual conversion
of a value. We can see that the $\iota$-rules for $\ccomp{c_1}{c_2}$
and $\carrow{c_1}{c_2}$ have an administrative rather than
computational role in that erasing the coercions (defined in the
expected way~\cite{Curien-Ghelli:MSCS92}) in the redex and in the
contractum of these rules leads to the same expression. It is worth
noting that terms of the form $\capp{(\carrow{c}{c})}{v}$ are
considered values, since they represent a coercion expecting another
value as argument (witness the last $\iota$-rule).

General contexts are closed terms with one hole (possibly under some
binders), and are ranged over by the metavariable $C$. We write
$\jctxtyping{\Gamma}{C}{\tau_1}{\tau_2}$ if for any $e$ with
$\jtyping{\Gamma}{e}{\tau_1}$ we have
$\jtyping{}{\inctx{C}{e}}{\tau_2}$. Contextual approximation, written
$\jctxapprox{\Gamma}{e_1}{e_2}{\tau}$, means that for any context $C$
and type $\tau'$, such that $\jctxtyping{\Gamma}{C}{\tau}{\tau'}$ if
$\inctx{C}{e_1}$ terminates, then so does $\inctx{C}{e_2}$. If
$\jctxapprox{\Gamma}{e_1}{e_2}{\tau}$ and
$\jctxapprox{\Gamma}{e_2}{e_1}{\tau}$, then we say that $e_1$ and
$e_2$ are contextually equivalent. It is this notion of program
equivalence that we take to express coherence of the coercion
semantics and characterize with logical relations later on in this
section.

%%%%%%%%%%%%%%%%%%%%%%%%
\subsubsection{Translation}
\label{subsubsec:translationSTLC}

\begin{figure}[!!t]
  \begin{mdframed}
\[
\begin{array}[t]{rcl}
\typtrans{\tarrow{\tau_1}{\tau_2}} & = & 
\tarrow{\typtrans{\tau_1}}{\typtrans{\tau_2}}
\\[2mm]
\typtrans{\tnat} & = & \tnat 
\\[2mm]
\typtrans{\ttop} & = & \tunit 
\end{array}
\]
\vspace{2mm}
\hrule
\vspace{2mm}
\[
\begin{array}[t]{rcl}
 \subtrans{\jsubtype{\tau}{\tau}}{\ruleSourceSRefl} & = & \cid \\[2mm]
\subtrans{\jsubtype{\tau}{\ttop}}{\ruleSourceSTop} & = & \ctop \\[2mm]
\subtrans{\jsubtype{\tau_1}{\tau_3}}{\ruleSourceSTrans(D_1, D_2)} 
    & = & \ccomp{\subtrans{\jsubtype{\tau_2}{\tau_3}}{D_1}}{
        \subtrans{\jsubtype{\tau_1}{\tau_2}}{D_2}} \\[2mm]
\subtrans{\jsubtype{\tarrow{\tau'_1}{\tau_1}}{\tarrow{\tau'_2}{\tau_2}}}{
    \ruleSourceSArrow(D_1, D_2)} & = &
    \carrow{\subtrans{\jsubtype{\tau'_2}{\tau'_1}}{D_1}}{
        \subtrans{\jsubtype{\tau_1}{\tau_2}}{D_2}}
\end{array}
\]
\vspace{2mm}
\hrule
\vspace{2mm}
\[
\begin{array}[t]{rcl}
\exptrans{x}{\ruleSourceTVar} & = & \var{x} \\[2mm]
\exptrans{\lam{x}{e}}{\ruleSourceTAbs(D)} & = & \lam{x}{\exptrans{e}{D}} \\[2mm]
\exptrans{\app{e_1}{e_2}}{\ruleSourceTApp(D_1, D_2)} & = &
\app{\exptrans{e_1}{D_1}}{\exptrans{e_2}{D_2}} \\[2mm]
\exptrans{\fix{f}{x}{e}}{\ruleSourceTFix(D)} & = &
    \fix{f}{x}{\exptrans{e}{D}} \\[2mm]
\exptrans{e}{\ruleSourceTSub(D_1,D_2)} & = &
    \capp{\subtrans{\jsubtype{\tau}{\tau'}}{D_2}}{\exptrans{e}{D_1}} \\[2mm]
\exptrans{\natconst{n}}{\ruleSourceTConst} & = & \natconst{n} 
\end{array}
\]
  \end{mdframed}
\caption{\label{fig:stlcCoercionSemantics}Coercion semantics for the
  $\lambda$-calculus with subtyping}
\end{figure}

The coercion semantics of the source language is given in
Figure~\ref{fig:stlcCoercionSemantics}. The function $\subtrans{.}{.}$
translates subtyping proofs into coercions, and function
$\exptrans{.}{.}$ translates typing derivations into terms of the
target language, whereas types are translated by the function
$\typtrans{.}$, which we extend to a point-wise translation of typing
environments. Both $\subtrans{.}{.}$ and $\exptrans{.}{.}$ are defined
by structural recursion on derivation trees, where the structure of
the tree $D$ is given by the second argument, consisting of the name
of the final rule in the derivation $D$ and the immediate subtrees of
$D$. For example, in the equation
\vspace{2mm}
\[
\begin{array}[t]{rcl}
\exptrans{e}{\ruleSourceTSub(D_1,D_2)} & = &
    \capp{\subtrans{\jsubtype{\tau}{\tau'}}{D_2}}{\exptrans{e}{D_1}}
\end{array}
\]

\vspace{2mm}\noindent
$\ruleSourceTSub(D_1,D_2)$ represents the tree
\vspace{2mm}
\begin{center}
  \AxiomC{$D_1$}
  \noLine
  \UnaryInfC{$\jtyping{\Gamma}{e}{\tau}$}
  \AxiomC{$D_2$}
  \noLine
  \UnaryInfC{$\jsubtype{\tau}{\tau'}$}
  \rulelabel{\ruleSourceTSub}
  \BinaryInfC{$\jtyping{\Gamma}{e}{\tau'}$}
  \DisplayProof
\end{center}

\vspace{2mm}\noindent
The translation functions themselves are rather straightforward; their
role is to replace subtyping derivations with coercions applied to
expressions being coerced from one type to another. The soundness of
the translation functions is ensured by the following lemma.
\newpage
\begin{lem}\label{lem:stlcTransPreservesTypes}
Coercion semantics preserves types.
\begin{enumerate}
\item If $\jtree{D}{\jsubtype{\tau_1}{\tau_2}}$ then 
    $\jctyping{\subtrans{\jsubtype{\tau_1}{\tau_2}}{D}}{
    \typtrans{\tau_1}}{\typtrans{\tau_2}}$.
\vspace{1mm}
\item If $\jtree{D}{\jtyping{\Gamma}{e}{\tau}}$ then
    $\jtyping{\envtrans{\Gamma}}{\exptrans{e}{D}}{\typtrans{\tau}}$.
\end{enumerate}
\end{lem}

The following example demonstrates the translation function and the
coercions at work.

\begin{exa}
Consider the program
$\app{(\lam{f}{\app{\var{f}}{\natconst{1}}})}{(\lam{x}{\var{x}})}$
in the source language.
Let $D$ be the derivation where variable $\var{f}$
has type $\tarrow{\tnat}{\ttop}$ and the type for expression
$\lam{x}{\var{x}}$ is derived in the following way:
\vspace{2mm}
\begin{prooftree}
\AxiomC{}
\rulelabel{\ruleSourceTVar}
\UnaryInfC{$\jtyping{\var{x}:\ttop}{\var{x}}{\ttop}$}
\rulelabel{\ruleSourceTAbs}
\UnaryInfC{$\jtyping{}{\lam{x}{\var{x}}}{\tarrow{\ttop}{\ttop}}$}
  \AxiomC{}
  \rulelabel{\ruleSourceSTop}
  \UnaryInfC{$\jsubtype{\tnat}{\ttop}$}
    \AxiomC{}
    \rulelabel{\ruleSourceSRefl}
    \UnaryInfC{$\jsubtype{\ttop}{\ttop}$}
  \rulelabel{\ruleSourceSArrow}
  \BinaryInfC{$\jsubtype{\tarrow{\ttop}{\ttop}}{\tarrow{\tnat}{\ttop}}$}
\rulelabel{\ruleSourceTSub}
\BinaryInfC{$\jtyping{}{\lam{x}{\var{x}}}{\tarrow{\tnat}{\ttop}}$}
\end{prooftree}
\vspace{2mm}
The coercion translation of such derivation puts a coercion application
in a place, where the subsumption rule was used in the type derivation:
\vspace{2mm}
\[
\exptrans{
    \app{(\lam{f}{\app{\var{f}}{\natconst{1}}})}{(\lam{x}{\var{x}})}
}{D} = 
    \app{(\lam{f}{\app{\var{f}}{\natconst{1}}})}{
        (\capp{(\carrow{\ctop}{\cid})}{(\lam{x}{\var{x}})})}
\]

\vspace{2mm}\noindent
The result of the translation can be reduced using $\beta$ and
$\iota$-reductions:
\vspace{2mm}
\[
\begin{array}{ll}
\app{(\lam{f}{\app{\var{f}}{\natconst{1}}})}{
   (\capp{(\carrow{\ctop}{\cid})}{(\lam{x}{\var{x}})})}
& \rawredi{\beta}
\\[1mm]
  \app{\capp{(\carrow{\ctop}{\cid})}{(\lam{x}{\var{x}})}}{\natconst{1}}
& \rawredi{\iota}
\\[1mm]
  \capp{\cid}{(\app{(\lam{x}{\var{x}})}{(\capp{\ctop}{\natconst{1}})})}
& \rawredi{\iota}
\\[1mm]
  \capp{\cid}{(\app{(\lam{x}{\var{x}})}{\vunit})}
& \rawredi{\beta}
\\[1mm]
  \capp{\cid}{\vunit}
& \rawredi{\iota}
\\[1mm]
  \vunit
&
\end{array}
\]

\vspace{2mm}\noindent
First, we perform $\beta$-reduction, since
$\capp{(\carrow{\ctop}{\cid})}{(\lam{x}{\var{x}})}$ is a value.
Thereafter, the arrow coercion $(\carrow{\ctop}{\cid})$ gets two
arguments, so it can be $\iota$-reduced by distributing coercions
$\ctop$ and $\cid$ between the argument and the result of the identity
function. Then we continue the reduction using the call-by-value
strategy. Note that both $\beta$- and $\iota$-reductions are needed
during the evaluation.
\qed
\end{exa}

The problem of coherence is illustrated in the next example.
\begin{exa}
  \label{ex:hetero}
Coercion semantics can produce distant results for different typing
derivations, even in such simple calculus as presented in this
section. Consider the fixed-point operator
$\fix{y}{f}{\lam{x}{\app{\app{\var{f}}{(\app{\var{y}}{\var{f}})}}{\var{x}}}}$
expressed using recursive functions.  Assuming
$\jsubtype{\tau}{\tau'}$, one possible type of such an expression is
$\tarrow{(\tarrow{(\tarrow{\tau}{\tau'})}{\tarrow{\tau}{\tau}})}
{\tarrow{\tau}{\tau}}$.  Let $D_1$ be a derivation where
variable~$\var{y}$ has the same type as whole expression, and we
coerce only the subexpression $(\app{\var{y}}{\var{f}})$ from type
$\tarrow{\tau}{\tau}$ to $\tarrow{\tau}{\tau'}$. On the other hand,
let $D_2$ be a derivation where the whole expression is coerced from
the type
$\tarrow{(\tarrow{(\tarrow{\tau}{\tau})}{\tarrow{\tau}{\tau}})}
{\tarrow{\tau}{\tau}}$, which is derived directly.  As a result of the
coercion semantics for the derivations $D_1$ and $D_2$ we get the
following programs in the target calculus:
\vspace{2mm}
\begin{eqnarray*}
e_1 & := & \exptrans{
  \fix{y}{f}{\lam{x}{\app{\app{\var{f}}{(\app{\var{y}}{\var{f}})}}{\var{x}}}}
}{D_1} = 
  \fix{y}{f}{\lam{x}{\app{\app{\var{f}}{
    (\capp{(\carrow{\cid}{c})}{(\app{\var{y}}{\var{f}})})
  }}{\var{x}}}} \\[1mm]
e_2 & := & \exptrans{
  \fix{y}{f}{\lam{x}{\app{\app{\var{f}}{(\app{\var{y}}{\var{f}})}}{\var{x}}}}
}{D_2} =
  \capp{(\carrow{(\carrow{(\carrow{\cid}{c})}{\cid})}{\cid})}
  {(\fix{y}{f}{\lam{x}{\app{\app{\var{f}}{(\app{\var{y}}{\var{f}})}}{\var{x}}}})}
\mathrm{,}
\end{eqnarray*}

\vspace{2mm}\noindent
where $\jctyping{c}{\tau}{\tau'}$ is a result of translating the proof
of $\jsubtype{\tau}{\tau'}$.  These terms are different values and it
is hard to find any reasonable equational theory to equate them.
However, as a consequence of next sections, they are contextually
equivalent.  Indeed, they exhibit similar behavior when applied to two
values $f$ and $v$.  We can perform three $\beta$-reductions starting
from the term $\app{\app{e_1}{f}}{v}$.
\vspace{2mm}
\[
\begin{array}{ll}
\app{\app{e_1}{f}}{v} & \rawredi{\beta}^2
\\[1mm]
  \app{\app{f}{(\capp{(\carrow{\cid}{c})}{(
    \app{e_1}{f}
  )})}}{v}
& \rawredi{\beta}
\\[1mm]
  \app{\app{f}{(\capp{(\carrow{\cid}{c})}{(
    \lam{x}{\app{\app{\var{f}}{
      (\capp{(\carrow{\cid}{c})}{(\app{e_1}{\var{f}})})
    }}{\var{x}}})})}
  }{v}
&
\end{array}
\]

\vspace{2mm}
\noindent
Reducing the term $\app{\app{e_2}{f}}{v}$ requires some extra $\iota$-reductions.
Let $e_0 =
\fix{y}{f}{\lam{x}{\app{\app{\var{f}}{(\app{\var{y}}{\var{f}})}}{\var{x}}}}$
and $f_0 = \capp{(\carrow{(\carrow{\cid}{c})}{\cid})}{f}$.
We have the following reduction path.
\vspace{2mm}
\[
\begin{array}{ll}
\app{\app{e_2}{f}}{v} & \rawredi{\iota}
\\[1mm]
  \app{\capp{\cid}{(
    \app{e_0}{f_0}
  )}}{v}
& \rawredi{\beta}
\\[1mm]
  \app{\capp{\cid}{(
    \lam{x}{\app{\app{f_0}{(\app{e_0}{f_0})}}{\var{x}}}
  )}}{v}
& \rawredi{\iota}\rawredi{\beta}
\\[1mm]
  \app{\app{f_0}{(\app{e_0}{f_0})}}{v}
& \rawredi{\beta}
\\[1mm]
  \app{\app{f_0}{(
    \lam{x}{\app{\app{f_0}{(\app{e_0}{f_0})}}{\var{x}}}
  )}}{v}
& \rawredi{\iota}
\\[1mm]
  \app{\capp{\cid}{(
    \app{f}{(
      \capp{(\carrow{\cid}{c})}{(
        \lam{x}{\app{\app{f_0}{(\app{e_0}{f_0})}}{\var{x}}}
      )}
    )}
  )}}{v} &
\end{array}
\]

\vspace{2mm}
\noindent
In both cases we obtained a term of the form
$\app{\app{f}{(\capp{(\carrow{\cid}{c})}{e})}}{v}$ (modulo
insignificant identity coercions), where $e$ is a result of applying
$e_i$ to $f$.  \qed
\end{exa}

%%%%%%%%%%%%%%%%%%%%%
\subsection{Logical relations}
\label{subsec:logicalRelationsSTLC}

In order to reason about contextual equivalence in the target
language, we define logical relations (Figure~\ref{fig:stlcLogRel}).
Relations are expressed in the LSLR logic described in
Section~\ref{sec:stepIndexedLogicalRelations}, so they are implicitly
step-indexed.

We call these relations heterogeneous because they are parameterized
by two types, one for each of the arguments. This property is
important for our coherence proof, since it makes it possible to
relate the results of the translation of two typing derivations which
assign different types to the same term, e.g., as in
Example~\ref{ex:hetero}. When both types $\tau_1$ and $\tau_2$ are
$\tnat$ or both are arrow types, the value relation
$\relV{\tau_1}{\tau_2}$ is standard. Two values are related for type
$\tnat$ if they are the same constant, and two functions are related
when they map related arguments to related results.  Because we have
many kinds of values representing functions, we follow Pitts and
Stark~\cite{Pitts-Stark:HOOTS98} in using an application for testing
functions, instead of a substitution (as in,
e.g.,~\cite{Dreyer-al:JFP12,Ahmed:ESOP06}).  The most interesting are
the cases when type parameters of the relation are different. When one
of these types is $\tunit$, then any values are in the relation,
because we do not expect them to carry any information---$\tunit$ is
the result of translating the $\ttop$ type.
In such a case we do not even require that related values are of the
kind described by their corresponding type.
We can do so since in the calculus each coercion applied to
a value is or reduces to a value. 
In calculi without this property we have to be more careful
(see Section \ref{subsubsec:multiBaseTypes}).
The logical relation is empty for different types which are not $\tunit$.

The relation $\relE{\tau_1}{\tau_2}$ for closed terms is defined by
biorthogonality. Two terms are related if they behave the same in
related contexts, and contexts are related (relation
$\relK{\tau_1}{\tau_2}$) if they yield the same observations when
plugged with related values. Yielding the same observations (relation
$\rawRelApprox$) is defined for each step-index separately:
$\relApprox{e_1}{e_2}$ is valid at $k$ iff termination of $e_1$ using
at most $k$ $\beta$-steps and any number of $\iota$-steps (written
$\stopsn{k}{e_1}$), implies termination of $e_2$ in any number of
$\beta$-steps and $\iota$-steps. This interpretation is monotone, so
the relation $\rawRelApprox$ can be added to the LSLR logic.

In order to extend the relation $\relE{\tau_1}{\tau_2}$ to open terms
we first define a relation $\relG{\Gamma_1}{\Gamma_2}$ on
substitutions (mapping variables to closed values) parameterized by a
pair of typing environments. Then we say that two open terms are
related (written
$\relEop{\Gamma_1}{\Gamma_2}{e_1}{e_2}{\tau_1}{\tau_2}$) when every
pair of related closing substitutions makes them related.

Notice that we do not assume that related terms have valid types.  Our
relations may include some ``garbage'', e.g.,
$(\natconst{1},\lam{x}{\var{x}})\in\relV{\tunit}{\tnat}$, but it is
non-problematic. One can mechanically prune these relations to
well-typed terms, but this change complicates formalization and we did
not find it useful.

\begin{figure}[!!t]
  \begin{mdframed}
\[
\begin{array}{rcl}
(v_1, v_2) \in \relV{\tnat}{\tnat} & \iff & \exists n, v_1 = v_2 = n \\[2mm]
(v_1, v_2) \in \relV{\tarrow{\tau'_1}{\tau_1}}{\tarrow{\tau'_2}{\tau_2}}
    & \iff & \forall (a_1, a_2) \in \relV{\tau'_1}{\tau'_2} .
    (\app{v_1}{a_1}, \app{v_2}{a_2}) \in \relE{\tau_1}{\tau_2} \\[2mm]
(v_1, v_2) \in \relV{\tunit}{\tau_2}  & \iff & \top\ \\[2mm]
(v_1, v_2) \in \relV{\tau_1}{\tunit}  & \iff & \top \\[2mm]
(v_1, v_2) \in \relV{\tau_1}{\tau_2} & \iff & \bot \qquad
    \textrm{otherwise} \\[4mm]
(e_1, e_2) \in \relE{\tau_1}{\tau_2}
    & \iff & \forall (E_1, E_2) \in \relK{\tau_1}{\tau_2} .
    \relApprox{\inctx{E_1}{e_1}}{\inctx{E_2}{e_2}} \\[4mm]
(E_1, E_2) \in \relK{\tau_1}{\tau_2}
    & \iff & \forall (v_1, v_2) \in \relV{\tau_1}{\tau_2} .
    \relApprox{\inctx{E_1}{v_1}}{\inctx{E_2}{v_2}} \\[4mm]
\validAt{k}{\relApprox{e_1}{e_2}} & \iff &
    \stopsn{k}{e_1} \implies \stops{e_2} \\[4mm]
(\gamma_1, \gamma_2) \in \relG{\Gamma_1}{\Gamma_2} & \iff &
    \forall x, (\gamma_1(x), \gamma_2(x)) \in \relV{\Gamma_1(x)}{\Gamma_2(x)} \\[4mm]
\relEop{\Gamma_1}{\Gamma_2}{e_1}{e_2}{\tau_1}{\tau_2} & \iff &
    \forall (\gamma_1, \gamma_2)\in\relG{\Gamma_1}{\Gamma_2} .
    (e_1\gamma_1, e_2\gamma_2) \in\relE{\tau_1}{\tau_2}
\end{array}
\]
  \end{mdframed}
\caption{\label{fig:stlcLogRel}Logical relations for the
  $\lambda$-calculus with explicit coercions}
\end{figure}

The relation $\rawRelApprox$ is preserved by reductions in the
following sense, where the third assertion expresses an elimination
rule of the later modality that is crucial in the subsequent proofs.

\begin{lem}\label{lem:stlcApproxRed}
The following assertions hold:
\begin{enumerate}
\item If $\iotared{e_1}{e'_1}$ and $\relApprox{e_1'}{e_2}$
  then $\relApprox{e_1}{e_2}$.
\vspace{1mm}
\item If $\iotared{e_2}{e'_2}$ and $\relApprox{e_1}{e'_2}$
    then $\relApprox{e_1}{e_2}$.
\vspace{1mm}
\item If $\betared{e_1}{e'_1}$ and $\later\relApprox{e_1'}{e_2}$
    then $\relApprox{e_1}{e_2}$.
\vspace{1mm}
\item If $\betared{e_2}{e'_2}$ and $\relApprox{e_1}{e'_2}$
    then $\relApprox{e_1}{e_2}$.
\end{enumerate}
\end{lem}

The proof of soundness of the logical relations follows closely the
standard technique for biorthogonal logical
relations~\cite{Pitts-Stark:HOOTS98,Dreyer-al:JFP12}. First, we need
to show compatibility lemmas, which state that the relation is
preserved by every language construct.

\begin{lem}[Compatibility]
The following assertions hold:
\begin{enumerate}
\item if $(x:\tau_1)\in\Gamma_1$ and $(x:\tau_2)\in\Gamma_2$
  then $\relEop{\Gamma_1}{\Gamma_2}{\var{x}}{\var{x}}{\tau_1}{\tau_2}$;
\vspace{1mm}
\item if
  $\relEop{(\eextend{\Gamma_1}{x}{\tau'_1})}{(\eextend{\Gamma_2}{x}{\tau'_2})}
    {e_1}{e_2}{\tau_1}{\tau_2}$ \\
  then
  $\relEop{\Gamma_1}{\Gamma_2}{\lam{x}{e_1}}{\lam{x}{e_2}}
    {\tarrow{\tau'_1}{\tau_1}}{\tarrow{\tau'_2}{\tau_2}}$;
\vspace{1mm}
\item if
  $\relEop{\Gamma_1}{\Gamma_2}{e_1}{e_2}
    {\tarrow{\tau'_1}{\tau_1}}{\tarrow{\tau'_2}{\tau_2}}$
  and
  $\relEop{\Gamma_1}{\Gamma_2}{e'_1}{e'_2}{\tau'_1}{\tau'_2}$ \\
  then
  $\relEop{\Gamma_1}{\Gamma_2}{\app{e_1}{e'_1}}{\app{e_2}{e'_2}}{\tau_1}{\tau_2}$;
\vspace{1mm}
\item if
  $\relEop
    {(\eextend{\eextend{\Gamma_1}{f}{\tarrow{\tau'_1}{\tau_1}}}{x}{\tau'_1})}
    {(\eextend{\eextend{\Gamma_2}{f}{\tarrow{\tau'_2}{\tau_2}}}{x}{\tau'_2})}
    {e_1}{e_2}{\tau_1}{\tau_2}$ \\
  then
  $\relEop{\Gamma_1}{\Gamma_2}{\fix{f}{x}{e_1}}{\fix{f}{x}{e_2}}
    {\tarrow{\tau'_1}{\tau_1}}{\tarrow{\tau'_2}{\tau_2}}$;
\vspace{1mm}
\item we have
  $\relEop{\Gamma_1}{\Gamma_2}{\natconst{n}}{\natconst{n}}{\tnat}{\tnat}$;
\vspace{1mm}
\item we have
  $\relEop{\Gamma_1}{\Gamma_2}{\vunit}{\vunit}{\tunit}{\tunit}$.
\end{enumerate}
\end{lem}
\begin{proof}
The proof is standard and directed by the definition of logical
relations.  We only show the proof for the case with recursive
functions, where step indexing simplifies reasoning.  Assume $\relEop
{(\eextend{\eextend{\Gamma_1}{f}{\tarrow{\tau'_1}{\tau_1}}}{x}{\tau'_1})}
{(\eextend{\eextend{\Gamma_2}{f}{\tarrow{\tau'_2}{\tau_2}}}{x}{\tau'_2})}
{e_1}{e_2}{\tau_1}{\tau_2}$ (*).  Since $\fix{f}{x}{e_1}$ and
$\fix{f}{x}{e_2}$ are values, it suffices to show that for every
substitutions $(\gamma_1, \gamma_2) \in \relG{\Gamma_1}{\Gamma_2}$ we
have $(\fix{f}{x}{e_1\gamma_1}, \fix{f}{x}{e_2}{\gamma_2}) \in
\relV{\tarrow{\tau'_1}{\tau_1}}{\tarrow{\tau'_2}{\tau_2}}$.  Now, we
use the L\"ob rule to assume the induction hypothesis\footnote{This
  reasoning step corresponds to the induction on indices.}
$(\fix{f}{x}{e_1\gamma_1}, \fix{f}{x}{e_2}{\gamma_2}) \in
\later\relV{\tarrow{\tau'_1}{\tau_1}}{\tarrow{\tau'_2}{\tau_2}}$ (**).
Unfolding the definition of the relation
$\relV{\tarrow{\tau'_1}{\tau_1}}{\tarrow{\tau'_2}{\tau_2}}$, we need
to show that for every $(v_1, v_2)\in\relV{\tau'_1}{\tau'_2}$ and
$(E_1, E_2)\in\relK{\tau_1}{\tau_2}$, we have $\relApprox
{\inctx{E_1}{\app{(\fix{f}{x}{e_1\gamma_1})}{v_1}}}
{\inctx{E_2}{\app{(\fix{f}{x}{e_2\gamma_2})}{v_2}}} $.  By
Lemma~\ref{lem:stlcApproxRed} (used twice), it suffices to prove that
$\later\relApprox
{\inctx{E_1}{\bisubst{e_1\gamma_1}{f}{\fix{f}{x}{e_1\gamma_1}}{x}{v_1}}}
{\inctx{E_2}{\bisubst{e_2\gamma_2}{f}{\fix{f}{x}{e_2\gamma_2}}{x}{v_2}}}
$.  Using the later introduction rule, we can remove the later
operator both in the goal and in the assumption (**). Now, we can show
that the substitutions
$\bisubst{\gamma_1}{f}{\fix{f}{x}{e_1\gamma_1}}{x}{v_1}$ and
$\bisubst{\gamma_2}{f}{\fix{f}{x}{e_2\gamma_2}}{x}{v_2}$ are related,
hence using (*) we conclude the proof.
\end{proof}

The only compatibility lemma specific to our relations is the lemma
for coercion application. Since the subsumption rule is not
syntax-directed, we expect from the coercions to preserve the logical
relation, even when they are applied to only one of the related
expressions.

\begin{lem}[Coercion compatibility]\label{lem:stlcCoercionCompat}
The logical relation is preserved by coercion application:
\begin{enumerate}
\item If $\jctyping{c}{\tau_1}{\tau_2}$ and
    $\relEop{\Gamma_1}{\Gamma_2}{e_1}{e_2}{\tau_1}{\tau_0}$ then
    $\relEop{\Gamma_1}{\Gamma_2}{\capp{c}{e_1}}{e_2}{\tau_2}{\tau_0}$.
\vspace{1mm}
\item If $\jctyping{c}{\tau_1}{\tau_2}$ and
    $\relEop{\Gamma_1}{\Gamma_2}{e_1}{e_2}{\tau_0}{\tau_1}$ then
    $\relEop{\Gamma_1}{\Gamma_2}{e_1}{\capp{c}{e_2}}{\tau_0}{\tau_2}$.
\end{enumerate}
\end{lem}
\begin{proof}
We prove both cases by induction on the typing derivation of the coercion $c$.
\end{proof}

Compatibility lemmas allow us to show the fundamental property of the
logical relations, stating that the logical relation is reflexive for
well-typed terms.

\begin{thm}[Fundamental property]
If $\jtyping{\Gamma}{e}{\tau}$ then $\relEop{\Gamma}{\Gamma}{e}{e}{\tau}{\tau}$.
\end{thm}
\begin{proof}
By induction on the derivation $\jtyping{\Gamma}{e}{\tau}$.
In each case we apply the corresponding compatibility lemma.
\end{proof}

The soundness of the logical relations is a direct consequence
of the following properties:
\emph{precongruence} which says that the logical relation is preserved by
any well-typed context,
and \emph{adequacy} which says that related programs have
the same observable behavior.

\begin{lem}[Precongruence]\label{lem:stlcPrecong}
If $\jctxtyping{\Gamma}{C}{\tau}{\tau_0}$
and $\relEop{\Gamma}{\Gamma}{e_1}{e_2}{\tau}{\tau}$
then $(\inctx{C}{e_1},\inctx{C}{e_2})\in\relE{\tau_0}{\tau_0}$.
\end{lem}
\begin{proof}
By induction on the derivation of context typing, using the
appropriate compatibility lemma in each case. For contexts containing
subterms we also need the fundamental property. For the empty context
we use the empty substitution, since the empty substitutions are in
relation $\relG{\envempty}{\envempty}$.
\end{proof}

\begin{lem}[Adequacy]\label{lem:stlcAdequacy}
If $(e_1, e_2)\in\relE{\tau}{\tau}$ then $\relApprox{e_1}{e_2}$.
\end{lem}
\begin{proof}
Let us show $\relApprox{\inctx{\mtectx}{e_1}}{\inctx{\mtectx}{e_2}}$.
Using the assertion $(e_1, e_2)\in\relE{\tau}{\tau}$, 
it suffices to show $(\mtectx, \mtectx)\in\relK{\tau}{\tau}$,
which is trivial, since values always terminate.
\end{proof}

\begin{thm}[Soundness]\label{thm:stlcRelSoundness}
If $\validAt{k}{\relEop{\Gamma}{\Gamma}{e_1}{e_2}{\tau}{\tau}}$
holds for every $k$,
then $\jctxapprox{\Gamma}{e_1}{e_2}{\tau}$.
\end{thm}
\begin{proof}
Suppose $\jctxtyping{\Gamma}{C}{\tau}{\tau_0}$ and
$\stops{\inctx{C}{e_1}}$, we need to show $\stops{\inctx{C}{e_2}}$. By
Lemma~\ref{lem:stlcPrecong} and Lemma~\ref{lem:stlcAdequacy} we know
$\validAt{k}{\relApprox{\inctx{C}{e_1}}{\inctx{C}{e_2}}}$ for every
$k$. Taking $k$ to be the number of steps in which $\inctx{C}{e_1}$
terminates, we have that $\inctx{C}{e_2}$ also terminates, by the
definition of $\rawRelApprox$.
\end{proof}

%%%%%%%%%%%%%%%%%%%%%%%
\subsection{Coherence of the coercion semantics}
\label{subsec:coherenceSTLC}

Having established soundness of the logical relations, we are in a
position to prove the main coherence lemma, phrased in terms of the
logical relations, and the coherence theorem.

\begin{lem}\label{lem:stlcCoherenceLemma}
If $\jtree{D_i}{\jtyping{\Gamma_i}{e}{\tau_i}}$ for $i=1,2$
are two typing derivations for the same term $e$ of the source language,
then $\relEop{\envtrans{\Gamma_1}}{\envtrans{\Gamma_2}}{
    \exptrans{e}{D_1}}{\exptrans{e}{D_2}}{
    \typtrans{\tau_1}}{\typtrans{\tau_2}}$.
\end{lem}
\begin{proof}
The proof follows by induction on the structure of both derivations
$D_1$ and $D_2$. At least one of these derivations is decreased in
every case. When one of the derivations starts with the subsumption
rule ($\ruleSourceTSub$), we apply Lemma~\ref{lem:stlcCoercionCompat}.
The coercion that we get after the translation is well-typed by
Lemma~\ref{lem:stlcTransPreservesTypes}. In other cases we just apply
the appropriate compatibility lemma.
\end{proof}

\begin{thm}[Coherence]
If $D_1$ and $D_2$ are derivations of the same typing judgment 
$\jtyping{\Gamma}{e}{\tau}$, then 
$\jctxapprox{\envtrans{\Gamma}}{\exptrans{e}{D_1}}{\exptrans{e}{D_2}}{\typtrans{\tau}}$.
\end{thm}
\begin{proof}
Immediately from Lemma~\ref{lem:stlcCoherenceLemma} and 
Theorem~\ref{thm:stlcRelSoundness}.
\end{proof}

%%%%%%%%%%%%%%%%%%%%%%%%%
\subsection{Variants}
\label{subsec:variants}

In this section we briefly discuss some possible extensions of the
results presented so far.

%%%%%%%%%%%%%
\subsubsection{Coercions as $\lambda$-terms}
\label{subsubsec:coercionsAsTerms}

The coercion semantics described here translates the source language
into the language with explicit coercions. We chose coercions to be a
separate syntactic category, because we found it very convenient,
especially for proving Lemma~\ref{lem:stlcCoercionCompat}.  However,
one can define a coercion semantics which translates subtyping proofs
directly to $\lambda$-expressions. Our result can be easily extended
for such a translation. Let $\erase{e}$ be a term $e$ with all the
coercions replaced by the corresponding expressions. To prove that for
any contextually equivalent terms $e_1$ and $e_2$ in the language with
coercions, terms $\erase{e_1}$ and $\erase{e_2}$ are contextually
equivalent in the language without coercions, we need three simple
facts that can be easily verified:
\begin{enumerate}
\item every well-typed term in the language without coercions is
    well typed in the language with coercions,
\item term $e$ terminates iff $\erase{e}$ terminates,
\item if context $C$ does not contain coercions then
  $\inctx{C}{\erase{e}} = \erase{\inctx{C}{e}}$.
\end{enumerate}

%%%%%%%%%%%%%
\subsubsection{Multiple base types} 
\label{subsubsec:multiBaseTypes}

In this presentation we consider languages with only one base type.
Adding more base types and some subtyping between them will not change
the general shape of the proof, but defining logical relations for
such a case is a little trickier.

Let $\baseTpSet$ be a set of base types and $\baseSub$ be a subtyping
relation on them. Assume for every $b\in\baseTpSet$ we have set
$\baseVal{b}$ of constants of type $b$.  These constants are values in
both source and target calculi.  Additionally, for each
$b\baseSub{}b'$ we have a corresponding coercion $\baseCrc{b}{b'}$ and
a function $\baseCrcFun{b}{b'}\colon\baseVal{b}\to\baseVal{b'}$.  The
$\iota$-rule for a coercion $\baseCrc{b}{b'}$ is defined as follows:
if $v\in\baseVal{b}$, then
$\iotared{\inctx{E}{\capp{\baseCrc{b}{b'}}{v}}}
{\inctx{E}{\baseCrcFun{b}{b'}(v)}}$.

The coherence of coercion semantics requires coherence on base types.
More precisely, we assume the following properties:
\begin{enumerate}
\item relation $\baseSub$ is reflexive and transitive;
\item for each $b\in\baseTpSet$ the function
  $\baseCrcFun{b}{b}$ is an identity;
\item if $b_1\baseSub{}b_2$ and $b_2\baseSub{}b_3$ then
  $\baseCrcFun{b_2}{b_3}\circ\baseCrcFun{b_1}{b_2} = \baseCrcFun{b_1}{b_3}$.
\end{enumerate}

We would stipulate that two values
$v_1$ and $v_2$ are related for base types $b_1$ and $b_2$ iff for
every common supertype $b$ of $b_1$ and $b_2$, coercing $v_1$ and
$v_2$ to $b$ yields the same constant:
\vspace{2mm}
\[
(v_1, v_2)\in\relV{b_1}{b_2} \iff
  v_1\in\baseVal{b_1} \wedge v_2\in\baseVal{b_2} \wedge
  \left(\forall b.b_1\baseSub{}b \wedge b_2\baseSub{}b \Rightarrow
    \baseCrcFun{b_1}{b}(v_1) = \baseCrcFun{b_2}{b}(v_2)\right)
    \]

\vspace{2mm}\noindent
Note that since relation $\baseSub$ is reflexive, for $b_1=b_2$ this
definition yields the identity relation on values of a base type
$b_1$, the same as in Section~\ref{subsec:logicalRelationsSTLC}.

Moreover, we have to be more careful with defining the logical
relation for the $\tunit$ type. The proof of
Lemma~\ref{lem:stlcCoercionCompat} relies on the fact that for every
$\jctyping{c}{\tau_1}{\tau_2}$ and $(v_1,v_2)\in\relV{\tau}{\tau_1}$,
the expression $\capp{c}{v_2}$ either is or reduces to a value. To
ensure that property, the relation for $\tunit$ and base type $b$
should relate any value with any value of type $b$:
\vspace{2mm}
\begin{eqnarray*}
(v_1, v_2)\in\relV{\tunit}{b} & \iff v_2\in\baseVal{b} \\[1mm]
(v_1, v_2)\in\relV{b}{\tunit} & \iff v_1\in\baseVal{b}
\end{eqnarray*}

%%%%%%%%%%%%%%%%%%%%%%%%%%%%%%%%%%%%%%%%%%%%%%%%%
\section{Coherence of a CPS translation of control-effect subtyping}
\label{sec:coherenceEffectSubtyping}

In this section we show that the results presented in
Section~\ref{sec:introducingLogicalRelations} can be adapted to a
considerably more complex calculus---a calculus of delimited control
with control-effect subtyping~\cite{Materzok-Biernacki:ICFP11}.

%%%%%%%%%%%%%%%%%%%%%%%%%
\subsection{Delimited continuations, informally}
\label{subsec:delimitedContinuations}

Control operators for delimited continuations, introduced
independently by Felleisen~\cite{Felleisen:POPL88} and by Danvy and
Filinski~\cite{Danvy-Filinski:LFP90}, allow the programmer to delimit
the current context of computation and to abstract such a delimited
context as a first-class value. They have found numerous applications
(see, e.g.,~\cite{Biernacka-al:LMCS05} for a list), including
Filinski's result showing that all computational effects are
expressible in terms of the delimited-control operators
$\mathsf{shift}$ and $\mathsf{reset}$~\cite{Filinski:POPL94}.

The calculus of delimited control studied in this work is the
call-by-value $\lambda$-calculus extended with natural numbers,
recursion, and the control operators $\mathsf{shift_0}$ ($\rawshiftz$)
and $\mathsf{reset_0}$ ($\rawresetz$)---a variant of $\mathsf{shift}$
and $\mathsf{reset}$~\cite{Danvy-Filinski:LFP90}. These operators have
recently enjoyed an upsurge of interest due to their considerable
expressive power and connections with the
$\lambda\mu$-calculi~\cite{Materzok-Biernacki:ICFP11,Materzok-Biernacki:APLAS12,
  Materzok:CSL13,Downen-Ariola:JFP14,Downen-Ariola:ICFP14,MunchMaccagnoni:LICS14}. Both
the calculus and the coercion semantics we consider in the rest of the
article are based on the type system and the CPS translation
introduced by Materzok and the first
author~\cite{Materzok-Biernacki:ICFP11}.

We will define the semantics of the calculus by a CPS translation to a
target calculus endowed with a reduction semantics, but if we were to
directly give reduction rules for $\mathsf{shift_0}$ and
$\mathsf{reset_0}$, they would be~\cite{Materzok-Biernacki:ICFP11}:
\vspace{2mm}
\begin{eqnarray*}
\inctx{F}{\resetz{\inctx{E}{\shiftz{x}{e}}}} 
& \rawredi{} &
\inctx{F}{\subst{e}{x}{\lam{y}{\resetz{\inctx{E}{y}}}}}
\\[1mm]
\inctx{F}{\resetz{v}} 
& \rawredi{} &
\inctx{F}{v} 
\end{eqnarray*}

\vspace{2mm}\noindent
where $E$ is a pure call-by-value evaluation context representing the
current delimited continuation (delimited by $\rawresetz$ and captured
by $\rawshiftz$), and $F$ is a metacontext, i.e., a general evaluation
context that consists of a number of pure evaluation contexts
separated by control delimiters.

Let us consider a simple example
\vspace{2mm}
\[
\add
    {\natconst{1}}
    {\resetz
      {\add
        {\natconst{10}}
        {\shiftz
          {k}
          {\add
            {\natconst{100}}
            {\app
              {k}
              {(\app{k}{0})}}}}}}
    \]

\vspace{2mm}\noindent
that represents an arithmetic expression over natural numbers. Here is
how this expression is evaluated according to the reduction rules (we
assume the standard reduction rules for $+$ and the call-by-value
$\beta$-reduction):
\vspace{2mm}
\[
\begin{array}{llr}
\add
    {\natconst{1}}
    {\resetz
      {\add
        {\natconst{10}}
        {\shiftz
          {k}
          {\add
            {\natconst{100}}
            {\app
              {k}
              {(\app{k}{0})}}}}}}
& \rawredi{} 
& \quad\quad (1)
\\[1mm]
\add
    {\natconst{1}}
    {\add
      {\natconst{100}}
      {\app
        {(\lam{y}{\resetz{\add{\natconst{10}}{y}}})}
        {(\app{(\lam{y}{\resetz{\add{\natconst{10}}{y}}})}{0})}}}
& \rawredi{}^3
& \quad\quad (2)
\\[1mm]
\add
    {\natconst{1}}
    {\add
      {\natconst{100}}
      {\app
        {(\lam{y}{\resetz{\add{\natconst{10}}{y}}})}
        {10}}}
& \rawredi{}^3
& \quad\quad (3)
\\[1mm]
\add
    {\natconst{1}}
    {\add
      {\natconst{100}}
      {20}}
& \rawredi{}^2
& \quad\quad (4)
\\[1mm]
121
&
\end{array}
\]

\vspace{2mm}\noindent
In step (1) the delimited continuation
$\lam{y}{\resetz{\add{\natconst{10}}{y}}}$ is captured and substituted
for $k$. In step (2) the captured continuation is applied to $0$, and
the result of this application, the value 10, is returned---the
captured continuation is functional in that it is composed with the
remaining computation, rather than abortive as it would be the case
for $\mathsf{call/cc}$. In step (3) the captured continuation is
applied to the passed value, and again it returns a value to the
remaining computation that consists in simple arithmetic, carried out
in step (4).

In contrast to $\mathsf{shift}$, $\mathsf{shift_0}$ is a control
operator that can explore and reorganize an arbitrary portion of the
metacontext. Here is an example:
\vspace{2mm}
\[
\begin{array}{llr}
\resetz{
  \add
      {\natconst{1}}
      {\resetz
      {\mul
        {\natconst{10}}
        {\shiftz
          {k_1}
          {\shiftz
            {k_2}
            {\app
              {k_1}
              {(\app{k_2}{0})}}}}}}}
& \rawredi{} 
& \quad\quad (1)
\\[1mm]
\resetz{
  \add
      {\natconst{1}}
      {\shiftz
        {k_2}
        {\app
          {(\lam{y}{\resetz{\mul{\natconst{10}}{y}}})}
          {(\app{k_2}{0})}}}}
& \rawredi{}
& \quad\quad (2)
\\[1mm]
\app
    {(\lam{y}{\resetz{\mul{\natconst{10}}{y}}})}
    {(\app{(\lam{y}{\resetz{\add{\natconst{1}}{y}}})}{0})}
& \rawredi{}^3
& \quad\quad (3)
\\[1mm]
\app
    {(\lam{y}{\resetz{\mul{\natconst{10}}{y}}})}
    {\natconst{1}}
& \rawredi{}^3
& \quad\quad (4)
\\[1mm]
10
&
\end{array}
\]

\vspace{2mm}\noindent
In step (1) $k_1$ is bound to the captured continuation representing
multiplication by $10$. In step (2) $k_2$ is bound to the captured
continuation representing incrementation by $1$. In step (3) and (4) 0
is first incremented and the result is then multiplied by 10---the
order of these operations is reversed compared to their occurrence in
the initial expression, which is achieved by repeatedly shifting
delimited continuations in steps (1) and (2) and by composing them in
the desired order.

Expressive type systems for delimited continuations are built around
the idea that the type of an expression depends on the type of a
context in which the expression is
immersed~\cite{Danvy-Filinski:DIKU89,Biernacka-Biernacki:PPDP09}. For
example, the expression
\vspace{2mm}
\[
\add{42}{\shiftz{k}{k}}
\]

\vspace{2mm}\noindent is well typed in such systems.  Given a context
$E$ that can be plugged with a value of type $\tnat$ and that returns
a value of some type $\tau$, assuming that $E$ does not trigger
control effects when plugged with a value, the evaluation of this
expression would return a value of type $\tarrow{\tnat}{\tau}$. Then,
given a metacontext $F$ that expects a value of that type, the
expression
\vspace{2mm}
\[
\inctx{F}
      {\resetz
        {\inctx{E}
          {\add{42}{\shiftz{k}{k}}}}}
      \]

\vspace{2mm}\noindent would be well~typed. We observe that the answer
type $\tau$ of the context $E$ differs from $\tarrow{\tnat}{\tau}$,
the type expected by the metacontext $F$. Such answer-type
modification is characteristic of type systems \`a la Danvy and
Filinski~\cite{Danvy-Filinski:DIKU89} and is necessary to exploit the
expressive power of typed delimited-control
operators~\cite{Asai-Kameyama:APLAS07,Biernacka-Biernacki:PPDP09,
  Materzok-Biernacki:ICFP11}.

Since the control operator $\mathsf{shift_0}$ is allowed to explore
the metacontext arbitrarily deep, the type of the expression should
actually depend not only on the type of its nearest enclosing context,
but also on the types of the remaining contexts that form the
metacontext. For example, the type of the term
\vspace{2mm}
\[
\shiftz{k_1}{\shiftz{k_2}{\app{k_1}{(\app{k_2}{\natconst{42}})}}}
\]

\vspace{2mm}\noindent
would express that given a context $E_1$ expecting a value of type
$\tau$ and with answer type $\tau'$, and a context $E_2$ expecting a
value of type $\tnat$ and with answer type $\tau$, the type of the
expression
\vspace{2mm}
\[
\resetz
    {\inctx{E_2}
    {\resetz
      {\inctx{E_1}
             {\shiftz{k_1}{\shiftz{k_2}{\app{k_1}{(\app{k_2}{\natconst{42}})}}}}}}}
\]

\vspace{2mm}\noindent
is $\tau'$. In fact the types in this example could be more complex
and express, e.g., that both $E_1$ and $E_2$ are effectful.
 
The calculus considered in the rest of this article was built around
the idea of types describing the relevant portion of the metacontext,
where, under some conditions, an expression that imposes certain
requirements on the metacontext can be used with a metacontext of
which more is known or assumed~\cite{Materzok-Biernacki:ICFP11}. For
example, a pure expression such as the constant $\natconst{42}$ can be
plugged in a pure evaluation context expecting values of type $\tnat$,
but also in arbitrarily complex metacontexts that have the inner-most
context accepting values of type $\tnat$. Coercions between types
describing metacontexts are possible thanks to the subtyping relation
that lies at the heart of the calculus presented in
Section~\ref{subsec:calculusEffects}.

%%%%%%%%%%%%%%%%%%%%%%%%%%%%%%%%%%%%%%%%%%%%%%%%%
\subsection{The lambda calculus with delimited control
  and effect subtyping}
\label{subsec:calculusEffects}

\begin{figure}[!!t]
  \begin{mdframed}
\begin{tabularx}{\textwidth}{rcXr}
$\tau$ & $\bnfdef$ & $\tnat \bnfor \tarrow{\tau}{T}$
    & (pure types) \\[2mm]
$T$ & $\bnfdef$ & $\puretype{\tau} \bnfor \efftype{\tau}{T}{T}$
    & (types) \\[2mm]
$e$ & $\bnfdef$ & $\var{x} \bnfor \lam{x}{e} \bnfor \app{e}{e}
  \bnfor \fix{x}{x}{e}
  \bnfor \shiftz{x}{e} \bnfor \resetz{e}
  \bnfor \natconst{n}$
    & (expressions)
\end{tabularx}
\vspace{4mm}
\hrule
\vspace{4mm}
{
\begin{center}
    \AxiomC{\phantomPremise}
    \rulelabel{\ruleSourceSRefl}
    \UnaryInfC{$\jsubtype{T}{T}$}
    \DisplayProof
\quad
    \AxiomC{$\jsubtype{T_2}{T_3}$}
    \AxiomC{$\jsubtype{T_1}{T_2}$}
    \rulelabel{\ruleSourceSTrans}
    \BinaryInfC{$\jsubtype{T_1}{T_3}$}
    \DisplayProof
\quad
    \AxiomC{$\jsubtype{\puretype{\tau_2}}{\puretype{\tau_1}}$}
    \AxiomC{$\jsubtype{T_1}{T_2}$}
    \rulelabel{\ruleSourceSArrow}
    \BinaryInfC{$\jsubtype{\puretype{(\tarrow{\tau_1}{T_1})}}{
        \puretype{(\tarrow{\tau_2}{T_2})}}$}
    \DisplayProof
\\[5mm]
    \AxiomC{$\jsubtype{\puretype{\tau_1}}{\puretype{\tau_2}}$}
    \AxiomC{$\jsubtype{T_2}{T_1}$}
    \AxiomC{$\jsubtype{U_1}{U_2}$}
    \rulelabel{\ruleSourceSCons}
    \TrinaryInfC{$\jsubtype{\efftype{\tau_1}{T_1}{U_1}}{
        \efftype{\tau_2}{T_2}{U_2}}$}
    \DisplayProof
\quad
    \AxiomC{$\jsubtype{T_1}{T_2}$}
    \rulelabel{\ruleSourceSLift}
    \UnaryInfC{$\jsubtype{\puretype{\tau}}{\efftype{\tau}{T_1}{T_2}}$}
    \DisplayProof
\\[5mm]
    \AxiomC{$\jtyping{\Gamma}{e}{T}$}
    \AxiomC{$\jsubtype{T}{U}$}
    \rulelabel{\ruleSourceTSub}
    \BinaryInfC{$\jtyping{\Gamma}{e}{U}$}
    \DisplayProof
\quad
    \AxiomC{$(x:\tau)\in\Gamma$}
    \rulelabel{\ruleSourceTVar}
    \UnaryInfC{$\jtyping{\Gamma}{\var{x}}{\puretype{\tau}}$}
    \DisplayProof
\quad
    \AxiomC{$\jtyping{\eextend{\Gamma}{x}{\tau}}{e}{T}$}
    \rulelabel{\ruleSourceTAbs}
    \UnaryInfC{$\jtyping{\Gamma}{\lam{x}{e}}{\puretype{\tarrow{\tau}{T}}}$}
    \DisplayProof
\\[5mm]
    \AxiomC{$\jtyping{\Gamma}{e_1}{\puretype{\tarrow{\tau}{T}}}$}
    \AxiomC{$\jtyping{\Gamma}{e_2}{\puretype\tau}$}
    \rulelabel{\ruleSourceTPApp}
    \BinaryInfC{$\jtyping{\Gamma}{\app{e_1}{e_2}}{T}$}
    \DisplayProof
\\[5mm]
    \AxiomC{$\jtyping{\Gamma}{e_1}{
        \efftype{(\tarrow{\tau_2}{\efftype{\tau_1}{U_4}{U_3}})}{U_2}{U_1}}$}
    \AxiomC{$\jtyping{\Gamma}{e_2}{\efftype{\tau_2}{U_3}{U_2}}$}
    \rulelabel{\ruleSourceTApp}
    \BinaryInfC{$\jtyping{\Gamma}{\app{e_1}{e_2}}{\efftype{\tau_1}{U_4}{U_1}}$}
    \DisplayProof
\\[5mm]
    \AxiomC{$\jtyping{
        \eextend{\eextend{\Gamma}{f}{\tarrow{\tau}{T}}}{x}{\tau}
        }{e}{T}$}
    \rulelabel{\ruleSourceTFix}
    \UnaryInfC{$\jtyping{\Gamma}{\fix{f}{x}{e}}{\puretype{\tarrow{\tau}{T}}}$}
    \DisplayProof
\quad
    \AxiomC{\phantomPremise}
    \rulelabel{\ruleSourceTConst}
    \UnaryInfC{$\jtyping{\Gamma}{\natconst{n}}{\puretype{\tnat}}$}
    \DisplayProof
\\[5mm]
    \AxiomC{$\jtyping{\eextend{\Gamma}{x}{\tarrow{\tau}{T}}}{e}{U}$}
    \rulelabel{\ruleSourceTSft}
    \UnaryInfC{$\jtyping{\Gamma}{\shiftz{x}{e}}{\efftype{\tau}{T}{U}}$}
    \DisplayProof
\quad
    \AxiomC{$\jtyping{\Gamma}{e}{\efftype{\tau}{\puretype{\tau}}{T}}$}
    \rulelabel{\ruleSourceTRst}
    \UnaryInfC{$\jtyping{\Gamma}{\resetz{e}}{T}$}
    \DisplayProof
\end{center}}
\end{mdframed}
\caption{\label{fig:shiftSource}The source language---the
  $\lambda$-calculus with delimited control and effect subtyping}
\end{figure}

The syntax and typing rules of the calculus of delimited control are
shown in Figure~\ref{fig:shiftSource}. Our presentation differs
slightly from the original one~\cite{Materzok-Biernacki:ICFP11}, but
only in some inessential details, and the two type systems are equally
expressive. Types are either pure ($\tau$) or effect annotated
($\efftype{\tau}{T_1}{T_2}$). A type $\efftype{\tau}{T_1}{T_2}$
describes a computation of type $\tau$ that when run in a delimited
context with an answer type $T_1$, yields a computation described by
$T_2$. For instance, the expression
$\shiftz{k_1}{\shiftz{k_2}{\app{k_1}{(\app{k_2}{\natconst{42}})}}}$,
considered in the previous section, can be given type
$\efftype{\tnat}{\tnat}{(\efftype{\tnat}{\tnat}{\tnat})}$, whereas
$\shiftz{k}{\natconst{42}}$ can be given type
$\efftype{\tnat}{\efftype{\tnat}{\tnat}{\tnat}}{\tnat}$.

The calculus comprises the simply typed lambda calculus (rules
$\ruleSourceTSub$, $\ruleSourceTVar$, $\ruleSourceTAbs$,
$\ruleSourceTPApp$) with the standard subtyping rules
($\ruleSourceSRefl$, $\ruleSourceSTrans$, $\ruleSourceSArrow$),
general recursion ($\ruleSourceTFix$), natural numbers
($\ruleSourceTConst$), and the remaining rules that describe control
effects at the level of types. First, the rule $\ruleSourceTSft$
corresponds to the operational behavior of $\shiftz{x}{e}$: assuming
that $e$, possibly using a captured context $\resetz{E}$ of type
$\tarrow{\tau}{T}$, can be plugged into a metacontext $F$ of type $U$,
it is sound to use the whole expression with the metacontext
$\inctx{F}{\resetz{E}}$ of type $\efftype{\tau}{T}{U}$. Accordingly,
the rule $\ruleSourceTRst$ expresses that $\resetz{e}$ can be used in
a metacontext $F$ of type $T$ provided $e$ can be plugged in the
metacontext $\inctx{F}{\resetz{\hole}}$ of type
$\efftype{\tau}{\tau}{T}$. Then, the rule $\ruleSourceTApp$ describes
an effectful application $\app{e_1}{e_2}$, where each of the
computation $e_1$, $e_2$, and the application itself can manipulate
the metacontext. This is a standard rule found already in Danvy and
Filinski's type-and-effect system for $\mathsf{shift}$ and
$\mathsf{reset}$~\cite{Danvy-Filinski:DIKU89}, where it was derived
from the CPS semantics of these operators.

Finally, we have two rules governing the subtyping of effectful
computations, namely $\ruleSourceSCons$ and $\ruleSourceSLift$. The
rule $\ruleSourceSCons$ follows from the CPS interpretation of
delimited continuations---a type $\efftype{\tau}{T_1}{T_2}$ is
interpreted in CPS as $\teffarrow{\tau}{T_1}{T_2}$, where
$\Rightarrow$ means an effectful function space (see
Section~\ref{subsec:coercionSemanticsEffects}). So,
$\efftype{\tau_1}{T_1}{U_1}$ is a subtype of
$\efftype{\tau_2}{T_2}{U_2}$ when $\tarrow{\tau_2}{T_2}$ is a subtype
of $\tarrow{\tau_1}{T_1}$ (the argument type is, as always, treated
contravariantly), and $U_1$ is a subtype of $U_2$ (the result type is,
as always, treated covariantly). The rule $\ruleSourceSLift$ is more
interesting and it says that a pure computation can be considered
impure, provided the answer type of the inner-most context can be
coerced into the type of the rest of the metacontext. We have, e.g.,
$\jsubtype{\tnat}{\efftype{\tnat}{\tau}{\tau}}$ by using
$\ruleSourceSLift$, which combined with $\ruleSourceSCons$ also
implies, e.g.,
$\jsubtype{\efftype{\tnat}{\efftype{\tnat}{\tau}{\tau}}{\tau'}}{\efftype{\tnat}{\tnat}{\tau'}}$.

The following example illustrates some of the typing rules of the
type system.
\begin{exa}
  \label{ex:typing-deriv}
  Taking $T = \efftype{\tnat}{\tnat}{\tnat}$ and $\tau =
  \tarrow{\tnat}{T}$ as well as $\Gamma = x:\tarrow{\tnat}{\tnat},
  y:\tau, z:\tarrow{\tnat}{\tnat}$ and $\Delta =
  \eextend{\Gamma}{k}{\tarrow{\tnat}{\tnat}}$, we have the following
  derivation $D$:

\vspace{2mm}
\begin{prooftree}
  \AxiomC{\phantomPremise}
  \rulelabel{\ruleSourceTVar}
  \UnaryInfC{$\jtyping{\Gamma}{x}{\tarrow{\tnat}{\tnat}}$}
  \AxiomC{$D_1$}
  \noLine
  \UnaryInfC{$\jtyping{\Gamma}{y}{\efftype{\tau}{\tnat}{\tnat}}$}
  \AxiomC{$D_2$}
  \noLine
  \UnaryInfC{$\jtyping{\Gamma}{\shiftz{k}{\app{z}{(\app{k}{\natconst{42}})}}}{T}$}
  \rulelabel{\ruleSourceTApp}
  \BinaryInfC{$\jtyping{\Gamma}{\app{y}{\shiftz{k}{\app{z}{(\app{k}{\natconst{42}})}}}}{T}$}
  \rulelabel{\ruleSourceTRst}
  \UnaryInfC{$\jtyping{\Gamma}{\resetz{\app{y}{\shiftz{k}{\app{z}{(\app{k}{\natconst{42}})}}}}}{\tnat}$}
  \rulelabel{\ruleSourceTPApp}
  \BinaryInfC{$\jtyping{\Gamma}
    {\app{x}{\resetz{\app{y}{\shiftz{k}{\app{z}{(\app{k}{\natconst{42}})}}}}}}
    {\tnat}$}
\end{prooftree}

\vspace{1mm}\noindent
where $D_1$ is
\vspace{2mm}
\begin{prooftree}
  \AxiomC{\phantomPremise}
  \rulelabel{\ruleSourceTVar}
  \UnaryInfC{$\jtyping{\Gamma}{y}{\tau}$}
  \AxiomC{\phantomPremise}
  \rulelabel{\ruleSourceSRefl}
  \UnaryInfC{$\jsubtype{\tnat}{\tnat}$}
  \rulelabel{\ruleSourceSLift}
  \UnaryInfC{$\jsubtype{\tau}{\efftype{\tau}{\tnat}{\tnat}}$}
  \rulelabel{\ruleSourceTSub}
  \BinaryInfC{$\jtyping{\Gamma}{y}{\efftype{\tau}{\tnat}{\tnat}}$}
\end{prooftree}

\vspace{1mm}\noindent
and $D_2$ is
\vspace{2mm}
\begin{prooftree}
  \AxiomC{\phantomPremise}
  \rulelabel{\ruleSourceTVar}
  \UnaryInfC{$\jtyping{\Delta}{z}{\tarrow{\tnat}{\tnat}}$}
  \AxiomC{\phantomPremise}
  \rulelabel{\ruleSourceTVar}
  \UnaryInfC{$\jtyping{\Delta}{k}{\tarrow{\tnat}{\tnat}}$}
  \AxiomC{\phantomPremise}
  \rulelabel{\ruleSourceTConst}
  \UnaryInfC{$\jtyping{\Delta}{\natconst{42}}{\tnat}$}
  \rulelabel{\ruleSourceTPApp}
  \BinaryInfC{$\jtyping{\Delta}{\app{k}{\natconst{42}}}{\tnat}$}
  \rulelabel{\ruleSourceTPApp}
  \BinaryInfC{$\jtyping{\Delta}{\app{z}{(\app{k}{\natconst{42}})}}{\tnat}$}
  \rulelabel{\ruleSourceTSft}
  \UnaryInfC{$\jtyping{\Gamma}{\shiftz{k}{\app{z}{(\app{k}{\natconst{42}})}}}{T}$}
\end{prooftree}
\qed
\end{exa}

%%%%%%%%%%%%%%%%%%%%%%%%
\subsection{Coercion semantics: a type-directed selective CPS translation}
\label{subsec:coercionSemanticsEffects}

\begin{figure}[!!t]
\begin{mdframed}
\begin{tabularx}{\textwidth}{rcXr}
$\tau$ & $\bnfdef$ & $\tnat \bnfor \tarrow{\tau}{T}$
    & (pure types) \\[2mm]
$T$ & $\bnfdef$ & $\tau \bnfor \teffarrow{\tau}{T}{T}$
    & (types) \\[2mm]
$c$ & $\bnfdef$ & $\cid \bnfor \ccomp{c}{c} \bnfor \carrow{c}{c}
    \bnfor \clift{c} \bnfor \ccons{c}{c}{c}$
    & (coercions) \\[2mm]
$e$ & $\bnfdef$ & $\var{x} \bnfor \lam{x}{e} \bnfor \app{e}{e}
  \bnfor \capp{c}{e} \bnfor \fix{x}{x}{e}
  \bnfor \natconst{n}
  $
    & (expressions) \\[2mm]
$v$ & $\bnfdef$ & $\var{x} \bnfor \lam{x}{e} \bnfor \fix{x}{x}{e}
  \bnfor \capp{(\carrow{c}{c})}{v} \bnfor \capp{\clift{c}}{v}
  \bnfor \capp{(\ccons{c}{c}{c})}{v} \bnfor \natconst{n} $
    & (values) \\[2mm]
$E$ & $\bnfdef$ & $\mtectx \bnfor \argectx{E}{e} \bnfor \valectx{v}{E}
  \bnfor \crcectx{c}{E} $
    & (evaluation contexts)
\end{tabularx}
\vspace{4mm}
\hrule
\vspace{4mm}
{
\begin{center}
    \AxiomC{\phantomPremise}
    \rulelabel{\ruleTargetSRefl}
    \UnaryInfC{$\jctyping{\cid}{T}{T}$}
    \DisplayProof
\quad
    \AxiomC{$\jctyping{c_1}{T_2}{T_3}$}
    \AxiomC{$\jctyping{c_2}{T_1}{T_2}$}
    \rulelabel{\ruleTargetSTrans}
    \BinaryInfC{$\jctyping{\ccomp{c_1}{c_2}}{T_1}{T_3}$}
    \DisplayProof
\\[5mm]
    \AxiomC{$\jctyping{c_1}{\tau_2}{\tau_1}$}
    \AxiomC{$\jctyping{c_2}{T_1}{T_2}$}
    \rulelabel{\ruleTargetSArrow}
    \BinaryInfC{$\jctyping{\carrow{c_1}{c_2}}{(\tarrow{\tau_1}{T_1})}{
        (\tarrow{\tau_2}{T_2})}$}
    \DisplayProof
\quad
    \AxiomC{$\jctyping{c}{T_1}{T_2}$}
    \rulelabel{\ruleTargetSLift}
    \UnaryInfC{$\jctyping{\clift{c}}{\puretype{\tau}}{(\teffarrow{\tau}{T_1}{T_2})}$}
    \DisplayProof
\\[5mm]
    \AxiomC{$\jctyping{c}{\tau_1}{\tau_2}$}
    \AxiomC{$\jctyping{c_1}{T_2}{T_1}$}
    \AxiomC{$\jctyping{c_2}{U_1}{U_2}$}
    \rulelabel{\ruleTargetSCons}
    \TrinaryInfC{$\jctyping{\ccons{c}{c_1}{c_2}}{
        (\teffarrow{\tau_1}{T_1}{U_1})}{(\teffarrow{\tau_2}{T_2}{U_2})}$}
    \DisplayProof
\\[5mm]
    \AxiomC{\phantomPremise}
    \rulelabel{\ruleTargetTConst}
    \UnaryInfC{$\jtyping{\Gamma}{\natconst{n}}{\tnat}$}
    \DisplayProof
\quad
\AxiomC{$(x:\tau)\in\Gamma$}
    \rulelabel{\ruleTargetTVar}
    \UnaryInfC{$\jtyping{\Gamma}{\var{x}}{\tau}$}
    \DisplayProof
\quad
\AxiomC{$\jtyping{\eextend{\Gamma}{x}{\tau}}{e}{T}$}
    \rulelabel{\ruleTargetTAbs}
    \UnaryInfC{$\jtyping{\Gamma}{\lam{x}{e}}{\tarrow{\tau}{T}}$}
    \DisplayProof
\\[5mm]
    \AxiomC{$\jtyping{\Gamma}{e_1}{\tarrow{\tau}{T}}$}
    \AxiomC{$\jtyping{\Gamma}{e_2}{\tau}$}
    \rulelabel{\ruleTargetTApp}
    \BinaryInfC{$\jtyping{\Gamma}{\app{e_1}{e_2}}{T}$}
    \DisplayProof
\\[5mm]
    \AxiomC{$\jtyping{\eextend{\Gamma}{x}{\tarrow{\tau}{T}}}{e}{U}$}
    \rulelabel{\ruleTargetTKAbs}
    \UnaryInfC{$\jtyping{\Gamma}{\lam{x}{e}}{\teffarrow{\tau}{T}{U}}$}
    \DisplayProof
\\[5mm]
    \AxiomC{$\jtyping{\Gamma}{e}{\teffarrow{\tau}{T}{U}}$}
    \AxiomC{$\jtyping{\Gamma}{v}{\tarrow{\tau}{T}}$}
    \rulelabel{\ruleTargetTKApp}
    \BinaryInfC{$\jtyping{\Gamma}{\app{e}{v}}{U}$}
    \DisplayProof
\\[5mm]
    \AxiomC{$\jctyping{c}{T}{U}$}
    \AxiomC{$\jtyping{\Gamma}{e}{T}$}
    \rulelabel{\ruleTargetTCApp}
    \BinaryInfC{$\jtyping{\Gamma}{\capp{c}{e}}{U}$}
    \DisplayProof
\hspace{-1.5mm}
    \AxiomC{$\jtyping{
        \eextend{\eextend{\Gamma}{f}{\tarrow{\tau}{T}}}{x}{\tau}
        }{e}{T}$}
    \rulelabel{\ruleTargetTFix}
    \UnaryInfC{$\jtyping{\Gamma}{\fix{f}{x}{e}}{\tarrow{\tau}{T}}$}
    \DisplayProof
\end{center}}
\vspace{4mm}
\hrule
\vspace{4mm}
\begin{center}
\setlength{\tabcolsep}{-8pt}
\begin{tabular}{ll}
\setlength{\tabcolsep}{1pt}
\begin{tabular}[t]{rcl}
$\inctx{E}{\app{(\lam{x}{e})}{v}}$ & $\rawredi{\beta}$ & 
    $\inctx{E}{\subst{e}{x}{v}}$ \\[2mm]
$\inctx{E}{\app{(\fix{f}{x}{e})}{v}}$ & $\rawredi{\beta}$ &
    $\inctx{E}{\bisubst{e}{f}{\fix{f}{x}{e}}{x}{v}}$ \\[2mm]
$\inctx{E}{\app{\capp{\clift{c}}{v_1}}{v_2}}$ & $\rawredi{\beta}$ &
    $\inctx{E}{\capp{c}{(\app{v_2}{v_1})}}$
\end{tabular}
\setlength{\tabcolsep}{1pt}
\begin{tabular}[t]{rcl}
$\inctx{E}{\capp{\cid}{v}}$ & $\rawredi{\iota}$ & $\inctx{E}{v}$ \\[2mm]
$\inctx{E}{\capp{(\ccomp{c_1}{c_2})}{v}}$ & $\rawredi{\iota}$ &
    $\inctx{E}{\capp{c_1}{(\capp{c_2}{v})}}$ \\[2mm]
$\inctx{E}{\app{\capp{(\carrow{c_1}{c_2})}{v_1}}{v_2}}$ & $\rawredi{\iota}$ &
    $\inctx{E}{\capp{c_2}{(\app{v_1}{(\capp{c_1}{v_2})})}}$ \\[2mm]
$\inctx{E}{\app{\capp{(\ccons{c}{c_1}{c_2})}{v_1}}{v_2}}$ & $\rawredi{\iota}$ &
    $\inctx{E}{\capp{c_2}{(\app{v_1}{(\capp{(\carrow{c}{c_1})}{v_2})})}}$
\end{tabular}
\end{tabular}
\end{center}
\end{mdframed}    
\caption{\label{fig:shiftTarget}The target language---the
  $\lambda$-calculus with explicit coercions of control effects}
\end{figure}

The type structure of the source calculus has been used by Materzok
and the first author to define the semantics of well-typed expressions
by a selective CPS translation of typing derivations into the
call-by-value
$\lambda$-calculus~\cite{Materzok-Biernacki:ICFP11}. Their translation
can be seen as a coercion semantics of the source calculus that
introduces explicit coercions in the image of the translation and
leaves pure expressions in direct style. Such a semantics thus can
serve as a basis for implementing delimited continuations as a
fragment of a conventional functional language. However, one should
first make sure that it is coherent.

The target calculus that we present next differs from the one
considered in~\cite{Materzok-Biernacki:ICFP11} in that it contains a
separate syntactic category of coercions as well as a dedicated
function type for expressions in CPS.

%%%%%%%%%%%%%
\subsubsection{Target calculus}
\label{subsubsec:targetEffects}

The syntax and typing rules of the target language are presented in
Figure~\ref{fig:shiftTarget}. There are two kinds of arrow type: the
usual one $\tarrow{\tau}{T}$ for regular functions and the effectful
one $\teffarrow{\tau}{T}{U}$ for expressions in CPS. We make this
distinction to express the fact that the CPS translation (see
Figure~\ref{fig:shiftCoercionSemantics}) yields expressions with
strong restrictions on the occurrence of terms in CPS: they are never
passed as arguments (typing environment consists of only pure types)
and they can be applied only to values representing continuations
(witness the rule $\ruleTargetTKApp$). Furthermore, observe that in
general terms in CPS, i.e., of type $\teffarrow{\tau}{T}{U}$ expect a
bunch of (delimited) continuations to produce the final answer. For
example the type
$\teffarrow{\tnat}{\tnat}{\teffarrow{\tnat}{\tnat}{\tnat}}$ is
inhabited by the term
$\lam{k_1}{\lam{k_2}{\app{k_1}{(\app{k_2}{\natconst{42}})}}}$.

The syntactic category $c$ of coercions and the typing rules defining
judgment $\jctyping{c}{T_1}{T_2}$ are in one-to-one correspondence
with the subtyping rules in the source calculus, discussed in
Section~\ref{subsec:calculusEffects}. In particular the rule
$\ruleTargetSLift$ allows us to treat a pure computation of type
$\tau$ as an effectful one (i.e., in CPS) of type
$\teffarrow{\tau}{T_1}{T_2}$, provided the answer type $T_1$ of the
immediate continuation is a subtype of $T_2$, the type describing the
remaining continuations. 

Again, the operational semantics distinguishes between $\beta$-rules
and $\iota$-rules.  We classified the last $\beta$-rule as ``actual
computation'' because it does not only rearrange coercions. It
translates back a lifted value $v_1$ and applies to it a given
continuation $v_2$. This rule and the last $\iota$-rule reduce a
coerced value applied to a continuation, so terms of the form
$(\capp{\clift{c}}{v})$ and $(\capp{\ccons{c}{c}{c}}{v})$ are
considered values. Notice that these values have effectful types. We
extend the notion of $\iota$-reduction to evaluation contexts:
$\iotared{E_1}{E_2}$ holds iff
$\iotared{\inctx{E_1}{v}}{\inctx{E_2}{v}}$ for every value $v$.

As in Section~\ref{sec:introducingLogicalRelations}, the metavariable
$C$ ranges over general closed contexts. We also define typing of
general contexts $\jctxtyping{\Gamma}{C}{T}{T_0}$ as before. The
definition of contextual approximation is necessarily slightly weaker,
because we allow only contexts with pure answer type:
we have $\jctxapprox{\Gamma}{e_1}{e_2}{T}$ if for every
$\jctxtyping{\Gamma}{C}{T}{\tau}$ a termination of $\inctx{C}{e_1}$
implies a termination of $\inctx{C}{e_2}$. Indeed, an
expression that requires a continuation to trigger computation can
hardly be considered a complete program.

%%%%%%%%%%%%%
\subsubsection{Translation}
\label{subsubsec:translationEffects}

\begin{figure}[!!t]
  \begin{mdframed}
  \[
  \begin{array}[t]{c@{\hspace{20mm}}c}
    \begin{array}{rcl}
      \puretyptrans{\tnat} & = & \tnat \\[2mm]
      \puretyptrans{\tarrow{\tau}{T}} & = &
      \tarrow{\puretyptrans{\tau}}{\efftyptrans{T}}
    \end{array}
    &
    \begin{array}{rcl}
      \efftyptrans{\puretype{\tau}} & = & \puretyptrans{\tau} \\[2mm]
      \efftyptrans{\efftype{\tau}{T}{U}} & = & 
      \teffarrow{\puretyptrans{\tau}}{\efftyptrans{T}}{\efftyptrans{U}}
    \end{array}
  \end{array}
  \]
  \vspace{2mm}
  \hrule
  \vspace{2mm}
  \[
  \begin{array}{rcl}
    \subtrans{\jsubtype{T}{T}}{\ruleSourceSRefl} & = & \cid \\[2mm]
    \subtrans{\jsubtype{T_1}{T_3}}{\ruleSourceSTrans(D_1, D_2)} 
    & = & \ccomp{\subtrans{\jsubtype{T_2}{T_3}}{D_1}}{
      \subtrans{\jsubtype{T_1}{T_2}}{D_2}} \\[2mm]
    \subtrans{\jsubtype{\puretype{\tarrow{\tau_1}{T_1}}}{
        \puretype{\tarrow{\tau_2}{T_2}}}}{
      \ruleSourceSArrow(D_1, D_2)} & = &
    \carrow{\subtrans{\jsubtype{\puretype{\tau_2}}{\puretype{\tau_1}}}{D_1}}{
      \subtrans{\jsubtype{T_1}{T_2}}{D_2}} \\[2mm]
    \subtrans{\jsubtype{\puretype{\tau}}{\efftype{\tau}{T}{U}}}{
      \ruleSourceSLift(D)} & = &
    \clift{\subtrans{\jsubtype{T}{U}}{D}} \\[2mm]
    \subtrans{\jsubtype{\efftype{\tau_1}{T_1}{U_1}}{\efftype{\tau_2}{T_2}{U_2}}}{
      \ruleSourceSCons(D,D_1,D_2)} & = &
    \ccons{\subtrans{\jsubtype{\puretype{\tau_1}}{\puretype{\tau_2}}}{D}}{
      \subtrans{\jsubtype{T_2}{T_1}}{D_1}}{
      \subtrans{\jsubtype{U_1}{U_2}}{D_2}}
  \end{array}
  \]
  \vspace{2mm}
  \hrule
  \vspace{2mm}
  \[
  \begin{array}[t]{rcl}
    \exptrans{e}{\ruleSourceTSub(D_1,D_2)} & = &
    \capp{\subtrans{\jsubtype{T}{U}}{D_2}}{\exptrans{e}{D_1}} \\[2mm]
    \exptrans{\var{x}}{\ruleSourceTVar} & = & \var{x} \\[2mm]
    \exptrans{\lam{x}{e}}{\ruleSourceTAbs(D)} & = & \lam{x}{\exptrans{e}{D}} \\[2mm]
    \exptrans{\app{e_1}{e_2}}{\ruleSourceTPApp(D_1, D_2)} & = &
    \app{\exptrans{e_1}{D_1}}{\exptrans{e_2}{D_2}} \\[2mm]
    \exptrans{\app{e_1}{e_2}}{\ruleSourceTApp(D_1, D_2)} & = &
    \lam{k}{\app{\exptrans{e_1}{D_1}}{(\lam{f}{
          \app{\exptrans{e_2}{D_2}}{(\lam{x}{
              \app{\app{\var{f}}{\var{x}}}{\var{k}}})}})}}  \\[2mm]
    \exptrans{\fix{f}{x}{e}}{\ruleSourceTFix(D)} & = &
    \fix{f}{x}{\exptrans{e}{D}} \\[2mm]
    \exptrans{\natconst{n}}{\ruleSourceTConst} & = & \natconst{n} \\[2mm]
    \exptrans{\shiftz{x}{e}}{\ruleSourceTSft(D)} & = &
    \lam{x}{\exptrans{e}{D}} \\[2mm]
    \exptrans{\resetz{e}}{\ruleSourceTRst(D)} & = &
    \app{\exptrans{e}{D}}{(\lam{x}{\var{x}})}
  \end{array}
  \]
  \end{mdframed}
  \caption{\label{fig:shiftCoercionSemantics}Type-directed selective CPS translation}
\end{figure}

The coercion semantics of the source language is given by the
type-directed selective CPS translation presented in
Figure~\ref{fig:shiftCoercionSemantics}. The translation is selective
because it leaves terms of pure type in direct style---witness, e.g,
the equations for variable or pure application. Effectful applications
are translated according to Plotkin's call-by-value CPS
translation~\cite{Plotkin:TCS75}, whereas the translation of
$\mathsf{shift_0}$ and $\mathsf{reset_0}$ is surprisingly
straightforward---$\mathsf{shift_0}$ is turned into a
lambda-abstraction expecting a delimited continuation, and
$\mathsf{reset_0}$ is interpreted by providing its subexpression with
the reset delimited continuation, represented by the identity
function. The following example illustrates the main points of the CPS
translation.

\begin{exa}
  Let us consider the derivation $D$ given in
  Example~\ref{ex:typing-deriv}. We have
  \vspace{2mm}
  \[
\exptrans{\app{x}{\resetz{\app{y}{\shiftz{k}{\app{z}{(\app{k}{\natconst{42}})}}}}}}{D} =
\app{x}{(\app{(\lam{l}{\app{\capp{\clift{\cid}}{y}}{(\lam{f}{\app{(\lam{k}{\app{z}{(\app{k}{\natconst{42}})}})}{(\lam{u}{\app{\app{f}{u}}{l}})}})}})}{(\lam{v}{v})})}
\]

\vspace{2mm}\noindent
where
\begin{itemize}
\item the applications
  $\app{x}{\resetz{\app{y}{\shiftz{k}{\app{z}{(\app{k}{\natconst{42}})}}}}}$,
  $\app{z}{(\app{k}{\natconst{42}})}$, and $\app{k}{\natconst{42}}$
  are pure and, hence, stay in direct style through the translation;
\item the application
  $\app{y}{\shiftz{k}{\app{z}{(\app{k}{\natconst{42}})}}}$ is
  effectful, and therefore translated to CPS, with $y$ coerced to a
  continuation-expecting expression.
\end{itemize}
\qed
\end{exa}

The following lemma establishes the type soundness of the CPS translation.

\begin{lem}
Coercion semantics preserves types:
\begin{enumerate}
\item If $\jtree{D}{\jsubtype{T_1}{T_2}}$ then 
    $\jctyping{\subtrans{\jsubtype{T_1}{T_2}}{D}}{
    \typtrans{T_1}}{\typtrans{T_2}}$.
\vspace{1mm}
\item If $\jtree{D}{\jtyping{\Gamma}{e}{T}}$ then
    $\jtyping{\envtrans{\Gamma}}{\exptrans{e}{D}}{\typtrans{T}}$.
\end{enumerate}
\end{lem}

The problem of coherence of the CPS translation is demonstrated in the
following example.

\begin{exa}
\label{exa:twoDerivations}
Let us consider the term
$\app{(\fix{f}{x}{\app{\var{f}}{\var{x}}})}{\natconst{1}}$ in the
source language. We derive the type ${\efftype{\tnat}{T}{T}}$ for it
in two ways: let $D_1$ be the derivation
\vspace{2mm}
\begin{prooftree}
\AxiomC{$\vdots$}
\noLine
\UnaryInfC{$\jtyping{f:\tarrow{\tnat}{\efftype{\tnat}{T}{T}},x:\tnat}{
    \app{\var{f}}{\var{x}}}{\efftype{\tnat}{T}{T}}$}
\rulelabel{\ruleSourceTFix}
\UnaryInfC{$\jtyping{}{\fix{f}{x}{\app{\var{f}}{\var{x}}}}{
    \puretype{\tarrow{\tnat}{\efftype{\tnat}{T}{T}}}}$}
    \AxiomC{}
    \rulelabel{\ruleSourceTConst}
    \UnaryInfC{$\jtyping{}{\natconst{1}}{\puretype{\tnat}}$}
\rulelabel{\ruleSourceTPApp}
\BinaryInfC{$\jtyping{}{
    \app{(\fix{f}{x}{\app{\var{f}}{\var{x}}})}{\natconst{1}}}{
    \efftype{\tnat}{T}{T}}$}
\end{prooftree}

\vspace{2mm}\noindent
and $D_2$ be the derivation
\vspace{2mm}
\begin{prooftree}
\AxiomC{$\vdots$}
\noLine
\UnaryInfC{$\jtyping{}{\fix{f}{x}{\app{\var{f}}{\var{x}}}}{
    \puretype{\tarrow{\tnat}{\efftype{\tnat}{T}{T}}}}$}
\AxiomC{$\vdots$}
\rulelabel{\ruleSourceTSub}
\BinaryInfC{$\jtyping{}{\fix{f}{x}{\app{\var{f}}{\var{x}}}}{
    \efftype{(\tarrow{\tnat}{\efftype{\tnat}{T}{T}})}{T}{T}}$}
    \AxiomC{}
    \rulelabel{\ruleSourceTConst}
    \UnaryInfC{$\jtyping{}{\natconst{1}}{\puretype{\tnat}}$}
    \AxiomC{$\vdots$}
    \rulelabel{\ruleSourceTSub}
    \BinaryInfC{$\jtyping{}{\natconst{1}}{\efftype{\tnat}{T}{T}}$}
\rulelabel{\ruleSourceTApp}
\BinaryInfC{$\jtyping{}{
    \app{(\fix{f}{x}{\app{\var{f}}{\var{x}}})}{\natconst{1}}}{
    \efftype{\tnat}{T}{T}}$}
\end{prooftree}

\vspace{2mm}\noindent
Then we have
\vspace{2mm}
\begin{eqnarray*}
\exptrans{\app{(\fix{f}{x}{\app{\var{f}}{\var{x}}})}{\natconst{1}}}{D_1}
    & = & \app{(\fix{f}{x}{\app{\var{f}}{\var{x}}})}{\natconst{1}} \\[1mm]
\exptrans{\app{(\fix{f}{x}{\app{\var{f}}{\var{x}}})}{\natconst{1}}}{D_2}
    & = & \lam{k}{\app{
        \capp{\clift{\cid}}{(\fix{f}{x}{\app{\var{f}}{\var{x}}})}}{(\lam{g}{
        \app{\capp{\clift{\cid}}{\natconst{1}}}{(\lam{y}{
        \app{\app{\var{g}}{\var{y}}}{\var{k}}})}})}}
\end{eqnarray*}

\vspace{2mm}\noindent
We observe that the two terms are quite distinct: one is a diverging
expression, and the other is a lambda abstraction. However, the
results of the next two sections show that these two terms are
contextually equivalent, and so both typing derivations have
equivalent coercion semantics.
\qed
\end{exa}

%%%%%%%%%%%%%%%%%%%%%%%%%
\subsection{Logical relations}
\label{subsec:logicalRelationsEffects}

\begin{figure}[!!t]
  \begin{mdframed}
\[
\hspace{-2mm}
\begin{array}{rcl}
(v_1, v_2) \in \relV{\tnat}{\tnat} & \iff & \exists n, v_1 = v_2 = n \\[2mm]
(v_1, v_2) \in \relV{\tarrow{\tau_1}{T_1}}{\tarrow{\tau_2}{T_2}}
    & \iff & \forall (a_1, a_2) \in \relV{\tau_1}{\tau_2} .
    (\app{v_1}{a_1}, \app{v_2}{a_2}) \in \relE{T_1}{T_2} \\[2mm]
(v_1, v_2) \in \relV{\tau_1}{\tau_2} & \iff & \bot \qquad \textrm{otherwise} \\[4mm]
(e_1, e_2) \in \relE{T_1}{T_2}
    & \iff & \forall (E_1, E_2) \in \relK{T_1}{T_2} .
    \relApprox{\inctx{E_1}{e_1}}{\inctx{E_2}{e_2}} \\[4mm]
(E_1, E_2) \in \relK{\tau_1}{\tau_2}
    & \iff & \forall (v_1, v_2) \in \relV{\tau_1}{\tau_2} .
    \relApprox{\inctx{E_1}{v_1}}{\inctx{E_2}{v_2}} \\[2mm]
(E_1, E_2) \in \relK{\tau_1}{\teffarrow{\tau_2}{T_2}{U_2}}
    & \iff & \exists T, (E, \kappa) \in \relEV{\tau_1}{T}{\tau_2}{T_2}, \\[1mm]
    && \hspace{6mm} (E'_1, E'_2) \in \relK{T}{U_2} . \\[1mm]
    && \hspace{12mm} \rtcloiotared{E_1}{\ctxcomp{E'_1}{E}}
    \wedge \rtcloiotared{E_2}{\inctx{E'_2}{\app{\hole}{\kappa}}} \\[2mm]
(E_1, E_2) \in \relK{\teffarrow{\tau_1}{T_1}{U_1}}{\tau_2}
    & \iff & \exists T, (\kappa, E) \in
    \later{}\relVE{\tau_1}{T_1}{\tau_2}{T}, \\[1mm]
    && \hspace{6mm} (E'_1, E'_2) \in \later{}\relK{U_1}{T} . \\[1mm]
    && \hspace{12mm} \rtcloiotared{E_1}{\inctx{E'_1}{\app{\hole}{\kappa}}}
    \wedge \rtcloiotared{E_2}{\ctxcomp{E'_2}{E}} \\[2mm]
(E_1, E_2) \in \relK{\teffarrow{\tau_1}{T_1}{U_1}}{\\\teffarrow{\tau_2}{T_2}{U_2}}
    & \iff & \exists (\kappa_1, \kappa_2) \in
    \relV{\tarrow{\tau_1}{T_1}}{\tarrow{\tau_2}{T_2}}, \\[1mm]
    && \hspace{6mm} (E'_1, E'_2) \in \relK{U_1}{U_2} . \\[1mm]
    && \hspace{12mm} \rtcloiotared{E_1}{\inctx{E'_1}{\app{\hole}{\kappa_1}}}
    \wedge \rtcloiotared{E_2}{\inctx{E'_2}{\app{\hole}{\kappa_2}}} \\[4mm]
(E, \kappa) \in \relEV{\tau_1}{T_1}{\tau_2}{T_2}
    & \iff & \forall (a_1, a_2)\in\relV{\tau_1}{\tau_2}.
    (\inctx{E}{a_1}, \app{\kappa}{a_2})\in\relE{T_1}{T_2} \\[4mm]
(\kappa, E) \in \relVE{\tau_1}{T_1}{\tau_2}{T_2}
    & \iff & \forall (a_1, a_2)\in\relV{\tau_1}{\tau_2}.
    (\app{\kappa}{a_1}, \inctx{E}{a_2})\in\relE{T_1}{T_2} \\[4mm]
(\gamma_1, \gamma_2) \in \relG{\Gamma_1}{\Gamma_2} & \iff &
    \forall x. (\gamma_1(x), \gamma_2(x)) \in \relV{\Gamma_1(x)}{\Gamma_2(x)} \\[4mm]
\relEop{\Gamma_1}{\Gamma_2}{e_1}{e_2}{T_1}{T_2} & \iff &
    \forall (\gamma_1, \gamma_2)\in\relG{\Gamma_1}{\Gamma_2} .
    (e_1\gamma_1, e_2\gamma_2) \in\relE{T_1}{T_2}
\end{array}
\]
  \end{mdframed}
\caption{\label{fig:shiftLogRel}Logical relations for the
  $\lambda$-calculus with explicit coercions of control effects}
\end{figure}

The logical relations are defined in Figure~\ref{fig:shiftLogRel}. We
use the metavariable $\kappa$ to range over values that are meant to
represent continuations. The relation $\relV{\tau_1}{\tau_2}$ for pure
values and the relation $\relE{T_1}{T_2}$ for expressions are similar
to the relations defined in
Section~\ref{subsec:logicalRelationsSTLC}. All information about
control effects is captured in the relation $\relK{T_1}{T_2}$ for
contexts. If $T_1$ and $T_2$ are pure, then we proceed as usual with
biorthogonal logical relations: two contexts are related if they
behave the same way for related values, since pure computations can
interact with their context only by returning a value.

For impure types (of the form $\teffarrow{\tau}{T}{U}$) contexts
should be plugged with effectful expressions which expect a
continuation (represented as a function) to trigger computation.  A
context is able to provide such a continuation $\kappa$ if it can be
decomposed as an application of the hole to $\kappa$ and the rest of
the context.  In general it does not mean that the context has
necessarily the form $\inctx{E'}{\app{\hole}{\kappa}}$, but that it
can be $\iota$-reduced to such a form. For instance, context
$\inctx{E}{\app{(\capp{(\ccons{c}{c_1}{c_2})}{\hole})}{\kappa}}$ does
not have an application to $\kappa$ as the inner-most element, but
still applies plugged value to a continuation
$\capp{(\carrow{c}{c_1})}{\kappa}$ after one $\iota$-step.  The
logical relation for contexts of impure types (case
$\relK{\teffarrow{\tau_1}{T_1}{U_1}}{\teffarrow{\tau_2}{T_2}{U_2}}$)
relates two contexts iff they can be decomposed (using
$\iota$-reduction) as applications to related continuations in related
contexts.

The most interesting are the cases that relate pure and impure
contexts. As previously, the impure context should be decomposed to a
continuation $\kappa$ and the rest of the context. Then the pure
context should be decomposed in such a way that the continuation
$\kappa$ is related with some portion $E$ of the pure context. The
answer type of $E$ cannot be retrieved from the type of the initial
pure context, so we quantify over all possible types.
Unlike the logical relations for
parametricity~\cite{Ahmed:ESOP06,Ahmed-al:POPL09}
we quantify over syntactic types. In order to make
the construction well-founded, the relations are defined by nested
induction on step indices and on the structure of the second
type. Notice that step indices play a role only in one case---when we
quantify over the second type and the later operator guards the
non-structural use of the relations $\relVE{\tau_1}{T_1}{\tau_2}{T}$
and $\relK{U_1}{T}$.  The auxiliary relations
$\relEV{\tau_1}{T_1}{\tau_2}{T_2}$ and
$\relVE{\tau_1}{T_1}{\tau_2}{T_2}$ relate a portion of an evaluation
context with a value of an arrow type and they are defined analogously
to the value relation for functions.

The relations of this section possess properties analogous to the ones
of Section~\ref{subsec:logicalRelationsSTLC}, in particular the
relation $\rawRelApprox$ is preserved by reduction
(Lemma~\ref{lem:stlcApproxRed}) and the compatibility lemmas
(including Lemma~\ref{lem:stlcCoercionCompat}) hold. However, the
proof of the compatibility lemmas requires the following results that
establish the preservation of relations with respect to
$\iota$-reductions of evaluation contexts.

\begin{lem}
The following assertions hold:
\begin{enumerate}\label{lem:ectxIotaObs}
\item If $\rtcloiotared{E}{E'}$ and $\relApprox{\inctx{E'}{e_1}}{e_2}$
    then $\relApprox{\inctx{E}{e_1}}{e_2}$.
\vspace{1mm}
\item If $\rtcloiotared{E}{E'}$ and $\relApprox{e_1}{\inctx{E'}{e_2}}$
    then $\relApprox{e_1}{\inctx{E}{e_2}}$.
\vspace{1mm}
\item If $\rtcloiotared{E_1}{E'_1}$ and $(E'_1, E_2)\in\relK{T_1}{T_2}$
    then $(E_1, E_2)\in\relK{T_1}{T_2}$.
\vspace{1mm}
\item If $\rtcloiotared{E_2}{E'_2}$ and $(E_1, E'_2)\in\relK{T_1}{T_2}$
    then $(E_1, E_2)\in\relK{T_1}{T_2}$.
\end{enumerate}
\end{lem}

The rest of the soundness proof follows the same lines as in
Section~\ref{subsec:logicalRelationsSTLC}. Interestingly, the adequacy
lemma can be proved only for pure types, which is in harmony with the
notion of contextual equivalence in the target calculus.

\begin{thm}[Fundamental property]
If $\jtyping{\Gamma}{e}{T}$ then $\relEop{\Gamma}{\Gamma}{e}{e}{T}{T}$.
\end{thm}

\begin{lem}[Precongruence]
If $\jctxtyping{\Gamma}{C}{T}{\tau}$ and 
$\relEop{\Gamma}{\Gamma}{e_1}{e_2}{T}{T}$,
then $(\inctx{C}{e_1}, \inctx{C}{e_2})\in\relE{\tau}{\tau}$.
\end{lem}

\begin{lem}[Adequacy]
If $(e_1, e_2)\in\relE{\tau}{\tau}$ then $\relApprox{e_1}{e_2}$.
\end{lem}

\begin{thm}[Soundness]
If $\validAt{k}{\relEop{\Gamma}{\Gamma}{e_1}{e_2}{T}{T}}$ holds for every $k$,
then $\jctxapprox{\Gamma}{e_1}{e_2}{T}$.
\end{thm}

%%%%%%%%%%%%%%%%%%%%%%%%%
\subsection{Coherence of the CPS translation}
\label{subsec:coherenceEffects}

Although standard compatibility lemmas and coercion compatibility
suffice to prove soundness of logical relations, we need another kind
of compatibility to prove coherence, since there is another source of
ambiguity. Two typing derivations in the source language can be
different not only because of the subsumption rule, but also because
of two rules for application.

\begin{lem}[Mixed application compatibility]
The following assertions hold:
\begin{enumerate}
\item If $\relEop{\Gamma_1}{\Gamma_2}{f_1}{f_2}{
        \teffarrow{(\tarrow{\tau'_1}{\teffarrow{\tau_1}{U_4}{U_3}})}{U_2}{U_1}}{
        \tarrow{\tau'_2}{T_2}}$ \\
    and
    $\relEop{\Gamma_1}{\Gamma_2}{e_1}{e_2}{
        \teffarrow{\tau'_1}{U_3}{U_2}}{\tau'_2}$ \\
    then
    $\relEop{\Gamma_1}{\Gamma_2}{
        \lam{k}{\app{f_1}{(\lam{f}{
            \app{e_1}{(\lam{x}{
            \app{\app{\var{f}}{\var{x}}}{\var{k}}})}})}}}{\app{f_2}{e_2}}{
        \teffarrow{\tau'_1}{U_4}{U_1}}{T}$.
\vspace{1mm}
\item If $\relEop{\Gamma_1}{\Gamma_2}{f_1}{f_2}{\tarrow{\tau'_1}{T_1}}{
        \teffarrow{(\tarrow{\tau'_2}{\teffarrow{\tau_2}{U_4}{U_3}})}{U_2}{U_1}}$ \\
    and
    $\relEop{\Gamma_1}{\Gamma_2}{e_1}{e_2}{\tau'_1}{
        \teffarrow{\tau'_2}{U_3}{U_2}}$ \\
    then
    $\relEop{\Gamma_1}{\Gamma_2}{\app{f_1}{e_1}}{
        \lam{k}{\app{f_2}{(\lam{f}{
            \app{e_2}{(\lam{x}{
            \app{\app{\var{f}}{\var{x}}}{\var{k}}})}})}}}{T}{
        \teffarrow{\tau'_2}{U_4}{U_1}}$.
\end{enumerate}
\end{lem}
\begin{proof}
Both cases are similar, so we show only the first one.  We have to
show that both terms closed by substitutions have the same
observations in related contexts $(E_1,
E_2)\in\relK{\teffarrow{\tau'_1}{U_4}{U_1}}{T}$.  Since context $E_1$
is in relation for effectful type, by the definition of logical
relations and Lemma~\ref{lem:ectxIotaObs}, it can be decomposed as a
continuation $\kappa$ and the rest of the context. Now we have the
missing continuation $\kappa$ that can trigger computation in the
first term, so the rest of the proof consists in simple context
manipulations, applying definitions and performing reductions.
\end{proof}

\begin{lem}
If $\jtree{D_i}{\jtyping{\Gamma_i}{e}{T_i}}$ for $i=1,2$
are two typing judgments for the same term $e$ of the source language,
then $\relEop{\envtrans{\Gamma_1}}{\envtrans{\Gamma_2}}{
    \exptrans{e}{D_1}}{\exptrans{e}{D_2}}{
    \typtrans{T_1}}{\typtrans{T_2}}$.
\end{lem}

\begin{thm}[Coherence]
\label{thm:coherenceEffects}
If $D_1$ and $D_2$ are derivations of the same typing judgment
$\jtyping{\Gamma}{e}{T}$, then
$\jctxapprox{\envtrans{\Gamma}}{\exptrans{e}{D_1}}{\exptrans{e}{D_2}}{\typtrans{T}}$.
\end{thm}

%%%%%%%%%%%%%%%%%%%%%%%%%
\subsection{Coercions as \texorpdfstring{$\lambda$}{lambda}-terms}
\label{subsubsec:coercionsAsTermsEffects}

In contrast to the calculus considered in
Section~\ref{subsec:coherenceSTLC}, such a coherence theorem does not
imply coherence of the translation directly to the simply typed
$\lambda$-calculus (where coercions are expressed as
$\lambda$-terms). As a counterexample, consider the expression
$\app{(\fix{f}{x}{\app{\var{f}}{\var{x}}})}{\natconst{1}}$ and the two
derivations $D_1$ and $D_2$ presented in
Example~\ref{exa:twoDerivations}. Recall that
\vspace{2mm}
\begin{eqnarray*}
\exptrans{\app{(\fix{f}{x}{\app{\var{f}}{\var{x}}})}{\natconst{1}}}{D_1}
    & = & \app{(\fix{f}{x}{\app{\var{f}}{\var{x}}})}{\natconst{1}} \\[1mm]
\exptrans{\app{(\fix{f}{x}{\app{\var{f}}{\var{x}}})}{\natconst{1}}}{D_2}
    & = & \lam{k}{\app{
        \capp{\clift{\cid}}{(\fix{f}{x}{\app{\var{f}}{\var{x}}}})}{(\lam{g}{
        \app{\capp{\clift{\cid}}{\natconst{1}}}{(\lam{y}{
        \app{\app{\var{g}}{\var{y}}}{\var{k}}})}})}}
\end{eqnarray*}

\vspace{2mm}\noindent
The former term is a diverging computation, but the latter is
a lambda abstraction waiting for an argument (continuation).
After translation to simply typed $\lambda$-calculus,
these terms can be distinguished even by the context
$C = \app{(\lam{x}{\natconst{1}})}{\hole}$ with answer type $\tnat$.
But by Theorem~\ref{thm:coherenceEffects} these terms are equivalent.
This is because types in the target language
carry more information than simple types, and in
particular, an expression of a type $\teffarrow{\tau}{T}{U}$ is not a
usual function, but a computation waiting for a continuation, as
explained in Section~\ref{subsec:coercionSemanticsEffects}.
Computations cannot be passed as arguments, so the context $C$
is not well-typed in the target calculus.

But still we can prove some interesting properties of a direct
translation to the simply typed $\lambda$-calculus in two cases: when
control effects do not leak to the context or when we relate only
whole programs. Let $\erase{e}$ be a term $e$ with all coercions
replaced by corresponding expressions.

\begin{cor}
If $\jtree{D_1,D_2}{\jtyping{\Gamma}{e}{\tau}}$ and $\tau$ does not
contain any type of the form $\efftype{\tau'}{T}{U}$, then
$\erase{\exptrans{e}{D_1}}$ and $\erase{\exptrans{e}{D_2}}$ are
contextually equivalent.
\end{cor}

\begin{cor}
If $\jtree{D_1,D_2}{\jtyping{\Gamma}{e}{\tau}}$ then
$\erase{\exptrans{e}{D_1}}$ terminates iff $\erase{\exptrans{e}{D_2}}$
terminates.  Moreover, if $\tau=\tnat$ and one of the expressions
terminates to a constant, then the other term evaluates to the same
constant.
\end{cor}

%%%%%%%%%%%%%%%%%%%%%%%%%%%%%%%%%%%%%%%%%%%%%%%%%
\section{Coq formalization}
\label{sec:formalization}

%%%%%%%%%%%%%%%%%%%%%%%%%
\subsection{The library IxFree}
\label{subsec:IxFree}

Our Coq formalization accompanying this article is built on our IxFree
library that contains a shallow embedding of the LSLR logic similar to
Appel et al.'s formalization of the ``very modal
model''~\cite{Appel-POPL07} and Krebbers et al.'s Iris proof
mode~\cite{Krebbers-POPL17}.
Instead of using type \verb+Prop+ to
represent propositions, we use a special type of ``indexed
propositions'' defined as a type of monotone functions from \verb+nat+
to \verb+Prop+.

\begin{CoqCode}
\coqKW{Definition} monotone (P : nat $\to$ Prop) :=\
$\forall$ n, P (S n) $\to$ P n.
\coqKW{Definition} IProp := \UseVerb{CBrOpn} P : nat $\to$ Prop\
| monotone P \UseVerb{CBrCls}.

\coqKW{Definition} I_valid_at (n : nat) (P : IProp) :=\
proj1_sig P n.
\coqKW{Notation} "n $\models$ P" := (I_valid_at n P).
\end{CoqCode}

One of the main differences between our
library and Iris proof mode is a way of keeping track of
the assumptions. Instead of interpreting a sequent
$\varphi_1,\ldots,\varphi_n\vdash\psi$ directly, we treat it as
$k\models\psi$ with the standard Coq assumptions $k\models\varphi_1$,
\ldots, $k\models\varphi_n$. This approach is very convenient since it
allows for reusing a number of existing Coq tactics,
but it does not scale to e.g. linear logic like Iris.

Logical connectives including the later operator are functions on type
\verb+IProp+ with defined human readable notation.
The library provides lemmas and tactics representing the most
important inference rules. Tactics not only apply the corresponding
lemmas, but also hide the step index arithmetic from the user.
For instance, when proving the sequent $Q \vdash P \Rightarrow Q$
represented by the following Coq goal
\newpage
\begin{CoqCode}
P : IProp
Q : IProp
k : nat
H1 : k $\models$ Q
\coqGoal
k $\models$ P $\Rightarrow$ Q
\end{CoqCode}

\noindent
the introduction of implication tactic \verb+iintro H2+ behaves
exactly like introduction of implication rule, producing the goal

\begin{CoqCode}
P : IProp
Q : IProp
k : nat
H1 : k $\models$ Q
H2 : k $\models$ P
\coqGoal
k $\models$ Q
\end{CoqCode}

\noindent
even if the lemma corresponding to that rule 
requires quantification over all smaller indices:

\begin{CoqCode}
\coqKW{Lemma} I_arrow_intro\
\UseVerb{CBrOpn}n : nat\UseVerb{CBrCls}\
\UseVerb{CBrOpn}P Q : IProp\UseVerb{CBrCls} :
  ($\forall$ k, k $\leq$ n, (k $\models$ P) $\to$ (k $\models$ Q)) $\to$\
(n $\models$ P $\Rightarrow$ Q).
\end{CoqCode}

%%%%%%%%%%%%%%%%%%%%%%%%%
\subsection{Recursive predicates}
\label{subsec:recPredicates}

The LSLR logic allows for recursive predicates and relations, provided
all recursive occurrences are guarded by the later operator.  Such a
syntactic requirement is not compatible with structural recursion in
Coq, so we rely on the notion of
\emph{contractiveness}\cite{Appel-POPL07}.  Informally, a function is
contractive if it maps approximately equal arguments to more equal
results. This intuition can be expressed using the later modality:

\begin{CoqCode}
\coqKW{Definition} contractive (l : list Type) (f : IRel l $\to$ IRel l)
  : Prop := $\forall$ R$@1$ R$@2$, $\models$ $\later$(R$@1$ $\approx@i$ R$@2$)\
$\Rightarrow$ f R$@1$ $\approx@i$ f R$@2$.
\end{CoqCode}

\noindent
where \verb+IRel l+ is a type of indexed relations
on types described by \verb+l+,
and $\approx_i$ is an indexed version of relation equivalence.
The library provides a general method of constructing recursive relations 
as a fixed point of a contractive function:

\begin{CoqCode}
\coqKW{Definition} I_fix (l : list Type) (f : IRel l $\to$ IRel l) :
  contractive l f $\to$ IRel l.
\end{CoqCode}

\noindent
If all occurrences of the function argument are guarded by the later
operator, then the function can be proven to be contractive, and the proof
can be (mostly) automatized by the \verb+auto_contr+ tactic.

%%%%%%%%%%%%%%%%%%%%%%%%%%%%%%%%%%%%%%%%%%%%%%%%%
\section{Conclusion}
\label{sec:conclusion}

We have shown that the technique of logical relations can be used for
establishing the coherence of subtyping, when it is phrased in terms
of contextual equivalence in the target of the coercion
translation. In particular, we have demonstrated that a combination of
heterogeneity, biorthogonality and step-indexing provides a
sufficiently powerful tool for establishing coherence of effect
subtyping in a calculus of delimited control with the coercion
semantics given by a type-directed selective CPS
translation. Moreover, we have successfully applied the presented
approach also to other calculi with subtyping, e.g., as demonstrated
in this article for the simply-typed $\lambda$-calculus with
recursion. The Coq development accompanying this paper is based on a
new embedding of Dreyer et al.'s logic for reasoning about
step-indexing~\cite{Dreyer-al:LMCS11} that, we believe, considerably
improves the presentation and formalization of the logical relations.

Regarding logical relations for type-and-effect systems, there has
been work on proving correctness of a partial evaluator for
$\mathsf{shift}$ and $\mathsf{reset}$ by Asai~\cite{Asai:TFP05}, and
on termination of evaluation of the $\lambda$-calculi with
delimited-control operators by Biernacka et
al.~\cite{Biernacka-Biernacki:PPDP09,Biernacka-al:PPDP11} and by
Materzok and the first
author~\cite{Materzok-Biernacki:ICFP11}. Unsurprisingly, all these
results, like ours, are built on the notion of biorthogonality, even
if not mentioned explicitly. The distinctive feature of our
construction is a combination of heterogeneity and step-indexing that
supports reasoning about the observational equivalence of terms of
different types whose structure is very distant from each other, e.g.,
about direct-style and continuation-passing-style terms.

Logical relations presented in
Section~\ref{sec:coherenceEffectSubtyping} require step-indexing in
order to ensure well-formedness of the definition in the presence of
quantification over types. A similar problem occurs in polymorphic
$\lambda$-calculi and is usually resolved using quantification over
relations that describe semantic types. Adapting this approach to our
calculus is not straightforward, because we need quantification over a
single type, whereas the semantic types are defined for pairs of
types. An interesting question is if the step-indexing in the
relations of Section~\ref{sec:coherenceEffectSubtyping} can be avoided
by using quantification over relations.

The type systems considered in this work are monomorphic. It remains
to be investigated how the ideas presented in this article would carry
over to system F$_\leq$ and its extensions. In particular, we find it
worthwhile to develop a polymorphic type-and-effect system for
$\mathsf{shift_0}$, perhaps by marrying the type system of
Section~\ref{sec:coherenceEffectSubtyping} with Asai and Kameyama's
polymorphic type system for delimited
continuations~\cite{Asai-Kameyama:APLAS07}, along with a type-directed
selective CPS translation to system F with explicit
coercions. Establishing the coherence of the translation would again
be a crucial step in such a development.

%%%%%%%%%%%%%%%%%%%%%%%%%%%%%%%%%%%%%%%%%%%%%%%%%

\section*{Acknowledgments}
We thank Andr{\'e}s A. Aristiz{\'a}bal, Ma{\l}gorzata Biernacka, Klara
Zieli{\'n}ska, and the anonymous reviewers of TLCA 2015 and LMCS for
helpful comments on the presentation of this work.

%%%%%%%%%%%%%%%%%%%%%%%%%%%%%%%%%%%%%%%%%%%%%%%%%
%% Bibliography

\bibliographystyle{plain}

\begin{thebibliography}{10}

\bibitem{Ahmed-al:POPL09}
Amal Ahmed, Derek Dreyer, and Andreas Rossberg.
\newblock State-dependent representation independence.
\newblock In Benjamin~C. Pierce, editor, {\em Proceedings of the 36th Annual
  ACM Symposium on Principles of Programming Languages, POPL 2009}, pages
  340--353, Savannah, GA, USA, January 2009. ACM Press.

\bibitem{Ahmed:ESOP06}
Amal~J. Ahmed.
\newblock Step-indexed syntactic logical relations for recursive and quantified
  types.
\newblock In Peter Sestoft, editor, {\em Programming Languages and Systems,
  15th European Symposium on Programming, ESOP 2006}, volume 3924 of {\em
  Lecture Notes in Computer Science}, pages 69--83, Vienna, Austria, March
  2006. Springer.

\bibitem{Appel-McAllester:TOPLAS01}
Andrew~W. Appel and David McAllester.
\newblock An indexed model of recursive types for foundational proof-carrying
  code.
\newblock {\em ACM Transactions on Programming Languages and Systems},
  23(5):657--683, 2001.

\bibitem{Appel-POPL07}
Andrew~W. Appel, Paul{-}Andr{\'{e}} Melli{\`{e}}s, Christopher~D. Richards, and
  J{\'{e}}r{\^{o}}me Vouillon.
\newblock A very modal model of a modern, major, general type system.
\newblock In Matthias Felleisen, editor, {\em Proceedings of the 34th Annual
  ACM Symposium on Principles of Programming Languages, POPL 2007}, pages
  109--122, Nice, France, January 2007. ACM Press.

\bibitem{Asai:TFP05}
Kenichi Asai.
\newblock Logical relations for call-by-value delimited continuations.
\newblock In Marko van Eekelen, editor, {\em Proceedings of the 6th Symposium
  on Trends in Functional Programming (TFP 2005)}, pages 413--428, Tallinn,
  Estonia, September 2005. Institute of Cybernetics at Tallinn Technical
  University.

\bibitem{Asai-Kameyama:APLAS07}
Kenichi Asai and Yukiyoshi Kameyama.
\newblock Polymorphic delimited continuations.
\newblock In Zhong Shao, editor, {\em Proceedings of the 5th Asian Symposium on
  Programming Languages and Systems, APLAS'07}, volume 4807 of {\em Lecture
  Notes in Computer Science}, pages 239--254, Singapore, December 2007.

\bibitem{Barendregt:84}
Henk Barendregt.
\newblock {\em The Lambda Calculus: Its Syntax and Semantics}, volume 103 of
  {\em Studies in Logic and the Foundation of Mathematics}.
\newblock North-Holland, revised edition, 1984.

\bibitem{Biernacka-Biernacki:PPDP09}
Ma{\l}gorzata Biernacka and Dariusz Biernacki.
\newblock Context-based proofs of termination for typed delimited-control
  operators.
\newblock In Francisco~J. L{\'o}pez-Fraguas, editor, {\em Proceedings of the
  11th ACM-SIGPLAN International Conference on Principles and Practice of
  Declarative Programming (PPDP'09)}, pages 289--300, Coimbra, Portugal,
  September 2009. ACM Press.

\bibitem{Biernacka-al:LMCS05}
Ma{\l}gorzata Biernacka, Dariusz Biernacki, and Olivier Danvy.
\newblock An operational foundation for delimited continuations in the {CPS}
  hierarchy.
\newblock {\em Logical Methods in Computer Science}, 1(2:5):1--39, November
  2005.

\bibitem{Biernacka-al:PPDP11}
Ma{\l}gorzata Biernacka, Dariusz Biernacki, and Sergue{\"\i} Lenglet.
\newblock Typing control operators in the {CPS} hierarchy.
\newblock In Michael Hanus, editor, {\em Proceedings of the 13th ACM-SIGPLAN
  International Conference on Principles and Practice of Declarative
  Programming (PPDP'11)}, pages 149--160, Odense, Denmark, July 2011. ACM
  Press.

\bibitem{Biernacki-Polesiuk:TLCA15}
Dariusz Biernacki and Piotr Polesiuk.
\newblock Logical relations for coherence of effect subtyping.
\newblock In Thorsten Altenkirch, editor, {\em 13th International Conference on
  Typed Lambda Calculi and Applications (TLCA 2015)}, volume~38 of {\em Leibniz
  International Proceedings in Informatics (LIPIcs)}, pages 107--122, Warsaw,
  Poland, July 2015. Schloss Dagstuhl -- Leibniz-Zentrum fuer Informatik.

\bibitem{Breazu-Tannen-al:IaC91}
Val Breazu-Tannen, Thierry Coquand, Carl~A. Gunter, and Andre Scedrov.
\newblock Inheritance as implicit coercion.
\newblock {\em Information and Computation}, 93(1):172--221, 1991.

\bibitem{Curien-Ghelli:MSCS92}
Pierre{-}Louis Curien and Giorgio Ghelli.
\newblock Coherence of subsumption, minimum typing and type-checking in
  {F$_\leq$}.
\newblock {\em Mathematical Structures in Computer Science}, 2(1):55--91, 1992.

\bibitem{Danvy-Filinski:DIKU89}
Olivier Danvy and Andrzej Filinski.
\newblock A functional abstraction of typed contexts.
\newblock DIKU Rapport 89/12, DIKU, Computer Science Department, University of
  Copenhagen, Copenhagen, Denmark, July 1989.

\bibitem{Danvy-Filinski:LFP90}
Olivier Danvy and Andrzej Filinski.
\newblock Abstracting control.
\newblock In Mitchell Wand, editor, {\em Proceedings of the 1990 ACM Conference
  on Lisp and Functional Programming}, pages 151--160, Nice, France, June 1990.
  ACM Press.

\bibitem{Downen-Ariola:ICFP14}
Paul Downen and Zena~M. Ariola.
\newblock Compositional semantics for composable continuations: from abortive
  to delimited control.
\newblock In {\em Proceedings of the 19th {ACM} {SIGPLAN} International
  Conference on Functional Programming (ICFP'14), Gothenburg, Sweden, September
  1-3, 2014}, pages 109--122. ACM Press, 2014.

\bibitem{Downen-Ariola:JFP14}
Paul Downen and Zena~M. Ariola.
\newblock Delimited control and computational effects.
\newblock {\em Journal of Functional Programming}, 24(1):1--55, 2014.

\bibitem{Dreyer-al:LMCS11}
Derek Dreyer, Amal Ahmed, and Lars Birkedal.
\newblock Logical step-indexed logical relations.
\newblock {\em Logical Methods in Computer Science}, 7(2:16):1--37, 2011.

\bibitem{Dreyer-al:JFP12}
Derek Dreyer, Georg Neis, and Lars Birkedal.
\newblock The impact of higher-order state and control effects on local
  relational reasoning.
\newblock {\em Journal of Functional Programming}, 22(4-5):477--528, 2012.

\bibitem{Felleisen:POPL88}
Matthias Felleisen.
\newblock The theory and practice of first-class prompts.
\newblock In Jeanne Ferrante and Peter Mager, editors, {\em Proceedings of the
  15th Annual ACM Symposium on Principles of Programming Languages, POPL 1988},
  pages 180--190, San Diego, California, January 1988. ACM Press.

\bibitem{Filinski:POPL94}
Andrzej Filinski.
\newblock Representing monads.
\newblock In Hans-J. Boehm, editor, {\em Proceedings of the 21st Annual ACM
  Symposium on Principles of Programming Languages, POPL 1994}, pages 446--457,
  Portland, Oregon, January 1994. ACM Press.

\bibitem{Krebbers-POPL17}
Robbert Krebbers, Amin Timany, and Lars Birkedal.
\newblock Interactive proofs in higher-order concurrent separation logic.
\newblock In {\em Proceedings of the 44th {ACM} {SIGPLAN} Symposium on
  Principles of Programming Languages, {POPL} 2017}, pages 205--217, Paris,
  France, January 2017. ACM Press.

\bibitem{Krivine:APAL94}
Jean{-}Louis Krivine.
\newblock Classical logic, storage operators and second-order lambda-calculus.
\newblock {\em Annals of Pure and Applied Logic}, 68(1):53--78, 1994.

\bibitem{Materzok:CSL13}
Marek Materzok.
\newblock {Axiomatizing subtyped delimited continuations}.
\newblock In Simona Ronchi~Della Rocca, editor, {\em Computer Science Logic
  2013 (CSL 2013)}, volume~23 of {\em Leibniz International Proceedings in
  Informatics (LIPIcs)}, pages 521--539, Torino, Italy, Sep 2013. Schloss
  Dagstuhl -- Leibniz-Zentrum fuer Informatik.

\bibitem{Materzok-Biernacki:ICFP11}
Marek Materzok and Dariusz Biernacki.
\newblock Subtyping delimited continuations.
\newblock In Manuel M.~T. Chakravarty, Zhenjiang Hu, and Olivier Danvy,
  editors, {\em Proceedings of the 2011 ACM SIGPLAN International Conference on
  Functional Programming (ICFP'11)}, pages 81--93, Tokyo, Japan, September
  2011. ACM Press.

\bibitem{Materzok-Biernacki:APLAS12}
Marek Materzok and Dariusz Biernacki.
\newblock A dynamic interpretation of the {CPS} hierarchy.
\newblock In Ranjit Jhala and Atsushi Igarashi, editors, {\em Proceedings of
  the 10th Asian Symposium on Programming Languages and Systems, APLAS'12},
  volume 7705 of {\em Lecture Notes in Computer Science}, pages 296--311,
  Kyoto, Japan, December 2012. Springer.

\bibitem{Mitchell:96}
John~C. Mitchell.
\newblock {\em Foundations for Programming Languages}.
\newblock MIT Press, 1996.

\bibitem{JHMorris:PhD}
James~H. Morris.
\newblock {\em Lambda Calculus Models of Programming Languages}.
\newblock PhD thesis, Massachusets Institute of Technology, 1968.

\bibitem{MunchMaccagnoni:LICS14}
Guillaume Munch{-}Maccagnoni.
\newblock Formulae-as-types for an involutive negation.
\newblock In Thomas~A. Henzinger and Dale Miller, editors, {\em Joint Meeting
  of the 23rd {EACSL} Annual Conference on Computer Science Logic {(CSL)} and
  the 29th Annual {ACM/IEEE} Symposium on Logic in Computer Science (LICS),
  {CSL-LICS} '14}, pages 70:1--70:10, Vienna, Austria, July 2014. {ACM} Press.

\bibitem{Pitts-Stark:HOOTS98}
Andrew Pitts and Ian Stark.
\newblock Operational reasoning for functions with local state.
\newblock In Andrew Gordon and Andrew Pitts, editors, {\em Higher Order
  Operational Techniques in Semantics}, pages 227--273. Publications of the
  Newton Institute, Cambridge University Press, 1998.

\bibitem{Plotkin:TCS75}
Gordon~D. Plotkin.
\newblock Call-by-name, call-by-value and the $\lambda$-calculus.
\newblock {\em Theoretical Computer Science}, 1:125--159, 1975.

\bibitem{Reynolds:TACS91}
John~C. Reynolds.
\newblock The coherence of languages with intersection types.
\newblock In Takayasu Ito and Albert~R. Meyer, editors, {\em Theoretical
  Aspects of Computer Software, International Conference {TACS} '91, Sendai,
  Japan, September 24-27, 1991, Proceedings}, volume 526 of {\em Lecture Notes
  in Computer Science}, pages 675--700. Springer, 1991.

\bibitem{Rompf-al:ICFP09}
Tiark Rompf, Ingo Maier, and Martin Odersky.
\newblock Implementing first-class polymorphic delimited continuations by a
  type-directed selective {CPS}-transform.
\newblock In Andrew Tolmach, editor, {\em Proceedings of the 2009 ACM SIGPLAN
  International Conference on Functional Programming (ICFP'09)}, pages
  317--328, Edinburgh, UK, August 2009. ACM Press.

\bibitem{Schwinghammer:JFP09}
Jan Schwinghammer.
\newblock Coherence of subsumption for monadic types.
\newblock {\em Journal of Functional Programming}, 19(2):157--172, 2009.

\end{thebibliography}

%%%%%%%%%%%%%%%%%%%%%%%%%%%%%%%%%%%%%%%%%%%%%%%%%

\end{document}